\documentclass[journal]{IEEEtran}
\usepackage{graphicx} 
\usepackage{subfigure}
\usepackage{amsmath}
\newtheorem{myDef}{Definition}
\newtheorem{assumption}{Assumption}
\newtheorem{proposition}{Proposition}
\newtheorem{proof}{Proof}

\usepackage{caption}

\newcommand{\re}[1]{{\color[rgb]{0,0,0}#1}}

\newtheorem{theorem}{Theorem}
\usepackage{algorithm}
\usepackage{algorithmic}

\usepackage{amsfonts}
\usepackage{algorithm, algorithmic}
\usepackage{color}
\usepackage{booktabs}
\usepackage{array}
\usepackage{epstopdf} 
\usepackage{url}
\usepackage{cite}

\begin{document}

\title{Deep Learning-based Physical-Layer Secret Key Generation for FDD Systems}

\author{Xinwei~Zhang, \IEEEmembership{Student Member,~IEEE,}
        Guyue~Li, \IEEEmembership{Member,~IEEE,}
        Junqing~Zhang,
        Aiqun~Hu,~\IEEEmembership{Senior Member,~IEEE,}
        Zongyue~Hou,
	Bin Xiao,~\IEEEmembership{Senior Member,~IEEE}

\thanks{This work was supported in part by the National Natural Science Foundation of China under Grant 6217011510, 61801115 and 61941115, in part by the Zhishan Youth Scholar Program of SEU (3209012002A3), in part by Jiangsu key R \& D plan BE2019109. (Corresponding author: G. Li)}
\thanks{X.~Zhang, G.~Li, Z.~Hou are with the School of Cyber Science and Engineering, Southeast University, Nanjing 210096, China (e-mail: zxw1998@seu.edu.cn; guyuelee@seu.edu.cn; zyhou@seu.edu.cn).}
\thanks{J.~Zhang is with the Department of Electrical Engineering and Electronics, University of Liverpool, Liverpool L69 3GJ, U.K. (e-mail: junqing.zhang@liverpool.ac.uk).}
\thanks{A.~Hu is with the School of Information Science and Engineering, and National Mobile Communications Research Laboratory, Southeast University, Nanjing 210096, China (e-mail: aqhu@seu.edu.cn).}
\thanks{G.~Li and A.~Hu are also with the Purple Mountain Laboratories for Network and Communication Security, Nanjing 210096, China.}
\thanks{B.~Xiao is with the Department of Computing, The Hong Kong Polytechnic University, Hong Kong, China (e-mail: csbxiao@comp.polyu.edu.hk).}
}


\maketitle

\begin{abstract}
Physical-layer key generation (PKG) establishes cryptographic keys from highly correlated measurements of wireless channels, which relies on reciprocal channel characteristics between uplink and downlink, is a promising wireless security technique for Internet of Things (IoT). However, it is challenging to extract common features in frequency division duplexing (FDD) systems as uplink and downlink transmissions operate at different frequency bands whose channel frequency responses are not reciprocal any more. Existing PKG methods for FDD systems have many limitations, i.e., high overhead and security problems. This paper proposes a novel PKG scheme that uses the feature mapping function between different frequency bands obtained by deep learning to make two users generate highly similar channel features in FDD systems. In particular, this is the first time to apply deep learning for PKG in FDD systems. We first prove the existence of the band feature mapping function for a given environment and a feedforward network with a single hidden layer can approximate the mapping function. Then a Key Generation neural Network (KGNet) is proposed for reciprocal channel feature construction, and a key generation scheme based on the KGNet is also proposed. Numerical results verify the excellent performance of the KGNet-based key generation scheme in terms of randomness, key generation ratio, and key error rate. Besides, the overhead analysis shows that the method proposed in this paper can be used for resource-contrained IoT devices in FDD systems.

\end{abstract}


\IEEEpeerreviewmaketitle

\section{Introduction}

\IEEEPARstart{W}{ith} the rapid development of the fifth generation (5G) and beyond communication systems, the security of wireless communication has received increasing attention~\cite{zou2016survey}. 
Due to the open nature of wireless channels, attackers can initiate various attacks such as eavesdropping, which pose a huge threat to wireless security. Traditionally, the cryptographic approaches including symmetric key cryptography and asymmetric key cryptography have been used to protect confidential information from eavesdroppers \cite{william2019cryptography}. However,  due to the three typical characteristics of 5G Internet of Things (IoT) networks, concerning mobility, massive IoT with resource constraints and heterogeneous hierarchical architecture, the traditional cryptographic mechanisms face problems such as difficulty in key distribution, excessive reliance on mathematical complexity, etc., and are not suitable or efficient \cite{wang2019physical}. To mitigate these issues, researchers have developed a new secure communication method from the physical layer in wireless communication, termed as \textit{physical-layer key generation (PKG)}~\cite{10.1145/3140257,zhang2016key,zhang2020new}. This technique uses the inherent randomness of fading channels between two legitimate users, namely Alice and Bob, to generate keys without the need for a third party.

The realization of PKG depends on three unique propagation characteristics of electromagnetic wave, namely channel reciprocity, temporal variation and spatial decorrelation. Among them, channel reciprocity indicates that the same channel characteristics can be observed at both ends of the same link, which is the basis for key generation~\cite{zhang2016key,zhang2020new}. 
For time division duplexing (TDD) systems, both the uplink and downlink are in the same carrier frequency band, and the channel responses obtained by Alice and Bob are reciprocal.
However, for frequency division duplexing (FDD) systems, the uplink and downlink transmit over different carrier frequencies, and the uplink and downlink experience different fading. Hence, most of the mutually accessible channel parameters used in TDD systems, such as received signal strength, channel gain, envelope, and phase, may be completely different between the uplink and downlink in FDD systems \cite{Li2019physical}. Therefore, it is challenging to find  reciprocal channel features for key generation in FDD systems. On the other hand, FDD dominates existing cellular communications, such as LTE and narrowband IoT. Key generation for these FDD-based systems will provide information-theoretically secure keys for them, hence is strongly desirable.

There have been several key generation methods developed for FDD systems~\cite{wang2012wireless,liu2019secret,goldberg2013method,wu2013secret,qin2016exploiting,allam2017channel,li2018constructing}. Those methods generate keys by extracting frequency-independent reciprocal channel parameters or constructing reciprocal channel gains, which have many limitations, i.e., high overhead and security problems \cite{goldberg2013method}.  Therefore, how to design a key generation scheme in a secure manner at a low communication overhead for FDD systems is still an open question.

To address this open problem, we consider how to construct the reciprocal features used to generate the key in FDD systems firstly. In \cite{alrabeiah2019deep}, the channel-to-channel mapping function in frequency is proved to exist under the condition that the channel mapping of the candidate position to the antenna is bijective, and the condition is realized with high probability in several practical multiple-input multiple-output (MIMO) systems. Inspired by this, we reveal the existence of band feature mapping function for a given environment in a FDD orthogonal frequency-division multiplexing (OFDM) system, which means that the feature mapping can be used to construct reciprocal features in FDD systems. However, the feature mapping between different frequency bands is difficult to be described by mathematical formulas.
To solve this problem, we construct the reciprocal channel features by using deep learning to learn the channel mapping function between different frequency bands.
Then a new key generation method based on the band feature mapping is proposed for FDD systems. The features of one frequency band is estimated simultaneously by both Alice and Bob to generate the key. 
In particular, this is the first time to apply deep learning for secret key generation in FDD systems.
The main contributions of this paper are as follows.
\begin{itemize}
	\item We prove the existence of the channel feature mapping function between different frequency bands for a given environment under the condition that the mapping function from the candidate user positions to the channels is bijective. Then, we prove that the channel feature mapping between different frequency bands can be approximated by a feedforward network with a single hidden layer. The above conclusions prove the feasibility of using deep learning to construct reciprocity features, and provide a new solution for the PKG in FDD systems. 
	\item We propose a Key Generation neural Network (KGNet) for band feature mapping to generate reciprocal channel features, and verify the performance in the simulation. Compared with three benchmark deep learning networks, the performance of KGNet can achieve good fitting and generalization performance under low signal-noise ratio (SNR). In addition, we artificially add noise by combining the noise-free and 0 dB dataset at a certain size to construct the training dataset, which improves the robustness of KGNet under low SNR. 
	\item Based on the KGNet, we propose a novel secret key generation scheme for FDD systems, in which Alice and Bob both estimate the features of one frequency band without any loop-back. Numerical results verify the excellent performance of the KGNet-based key generation scheme in terms of randomness, key generation ratio, and key error rate. Besides, the overhead analysis shows that the method proposed in this paper can be used for resource-contrained IoT devices in FDD systems.
\end{itemize}

The remainder of this paper is structured as follows. Section \ref{Related work}
presents the related work. The system model for FDD systems is introduced in Section \ref{System Model}. In Section \ref{DL}, we prove the existence of the band feature mapping for a given environment under a condition that the mapping function from the candidate user positions to the channels is bijective, and a feedforward network with a single hidden layer can approximate the mapping. The channel feature construction algorithm based on KGNet is presented in Section \ref{KGNet}. Section \ref{KGNet-based KG} designs the key generation scheme. The simulation results for evaluating the performance of the KGNet and the proposed key generation scheme are provided in Section \ref{Simulation}, which is followed by conclusions in Section \ref{conclusion}.

\section{Related Work}
\label{Related work}
This section introduces related work, including the previous FDD key generation methods and the application of deep learning in the wireless physical layer.

\subsection{Secret Key Generation for FDD Systems}
In the past few years, several secret key generation methods have been developed for FDD systems. The main methods of FDD key generation are summarized as follows.
\begin{enumerate}
	\item Extracting the frequency-independent reciprocal channel parameters to generate the key \cite{wang2012wireless,liu2019secret}. In \cite{wang2012wireless}, the angle and delay are used to generate the key as they are supposed to hold the reciprocity in FDD systems. However, the accurate acquisition of the angle and delay requires a lot of resources, such as large bandwidth and multiple antennas \cite{vasisht2016eliminating}. In addition, a secret key generation method based on the reciprocity of channel covariance matrix eigenvalues is also proposed in FDD systems \cite{liu2019secret}. But this method requires special configuration of the antenna array.
	
	\item Establishing the channel with reciprocal channel gain by means of the additional reverse channel training phase to generate the key, called the loopback-based methods \cite{goldberg2013method,wu2013secret,qin2016exploiting,allam2017channel}. In \cite{goldberg2013method}, Alice and Bob generate the key by estimating the channel impulse response (CIR) of the combinatorial channel which is the combination of uplink and downlink channels. In \cite{wu2013secret}, instead of the CIR of the combinatory channel, only the uplink channel state information (CSI) is estimated by both sides of the communication using a special CSI probing method to generate the key. Furthermore, there are two secret key generation schemes, which generate secret key by exploiting shared physical channel information on non-reciprocal forward and reverse channels \cite{qin2016exploiting}. However, the loopback-based key generation methods require additional reverse channel training and multiple iterative interactions, which not only increase the complexity of channel detection, but also increase the risk of eavesdropping. Furthermore, it has been proved not security in~\cite{linning2018investigation}.
		
	\item Constructing reciprocal features based on the prior knowledge of the channel model by separating channel paths to generate keys \cite{li2018constructing}. However, in the complex multi-path environment, separating the channel paths is not simple.	
\end{enumerate}

Through the analysis of the previous FDD key generation methods, the existing key generation methods in the FDD systems have problems such as large overhead and insecurity. The key generation method proposed in this paper uses the mapping function between different frequency bands to construct reciprocal channel features from a new perspective. Alice and Bob only need to probe the channel once, without additional reverse channel training. In addition, we use deep learning to learn the mapping function between frequency bands, which is data-driven, so there are not many restrictions on the model, and no complicated calculations are needed to extract reciprocal channel parameters. 

\subsection{Deep Learning for Wireless Physical Layer}
Deep learning  has been introduced to the wireless physical layer and achieved excellent performance in many areas such as channel estimation \cite{yang2019deep1}, CSI feedback \cite{wen2018deep}, downlink CSI prediction \cite{wang2019ul}, modulation classification \cite{lin2020improved}, etc. 

Gao et al. \cite{gao2019deep} proposed a direct-input deep neural network (DI-DNN) to estimate channels by using the received signals of all antennas. 
Wen et al. \cite{wen2018deep} used deep learning technology to develop CSINet, a novel CSI sensing and recovery mechanism that learns to effectively use channel structure from training samples.
Yang et al. \cite{yang2019deep} used a spare complex-valued neural network (SCNet) for the downlink CSI prediction in FDD massive MIMO systems. Safari et al. \cite{safari2019deep} proposed a convolutional neural network (CNN) and a generative neural network (GAN) for predicting downlink CSI by observing the downlink CSI.
To the best of the authors' knowledge, there is no study on deep learning-based secret key generation method for FDD systems.

Alrabeiah et al. \cite{alrabeiah2019deep} used a fully-connected neural network to learn and approximate the channel-to-channel mapping function. Inspired by this, we propose KGNet to generate reciprocal channel features and propose a KGNet-based key generation scheme for FDD systems. The KGNet is able to learn the mapping function by off-line training. After training, KGNet is used to directly predict the reciprocal features. Therefore, the method proposed in this paper has a small overhead in practical application and great potential for practical deployment. In particular, This paper applies deep learning to key generation for the first time.

\section{System Model}
\label{System Model}

\subsection{Channel Model}
Key generation involves two legitimate users, namely Alice (one base station) and Bob (one user) as well as an eavesdropper, Eve, located $d$ m away from Bob. Alice and Bob will send signals to each other alternately and obtain channel estimation $\hat{h}_A$ and $\hat{h}_B$. In a TDD system, the channel estimation of Alice and Bob is highly correlated based on the channel reciprocity. Hence we can design a key generation protocol, $\mathcal{K}(\cdot)$, which will convert the analog channel estimation to a digital binary sequence.
The process can be expressed as
\begin{align}
	K_A &= \mathcal{K}(\hat{h}_A),\\
	K_B &= \mathcal{K}(\hat{h}_B).
\end{align}


The above protocol has worked well in TDD-based systems, e.g., WiFi \cite{mathur2008radio}, ZigBee \cite{aono2005wireless}, and LoRa \cite{xu2018lora,ruotsalainen2019experimental}. However, its adoption in FDD systems is extremely challenging, because the uplink and downlink are not reciprocal any more. For the first time, this paper will employ deep learning to construct reciprocal channel features at Alice and Bob and extend key generation for FDD-based systems.

Specifically, this paper considers Alice and Bob are equipped with a single antenna and operate at the FDD mode. Alice and Bob simultaneously transmit signals on different carrier frequencies, $f_{AB}$ and $f_{BA}$, respectively. We denote the links from Bob to Alice and from Alice to Bob as Band1 and Band2, respectively. The CIR consists of $N$ paths and can be defined as 
\begin{equation}
\begin{split}
h(f,\tau)=\sum_{n=0}^{N-1}\alpha_{n}e^{-j2\pi f\tau_n+j\phi_n}\delta(\tau-\tau_n),
\end{split}
\label{CIR}
\end{equation}
where $f$ is the carrier frequency, $\alpha_{n}$, $\tau_n$, and $\phi_n$ are the magnitude, delay, and phase shift of the $n^{th}$ path, respectively. Note that $\alpha_{n}$ depends on (i) the distance $d_n$ between Alice and Bob, (ii) the carrier frequency $f$, (iii) the scattering environment. The phase $\phi_n$ is determined by the scatterer(s) materials and wave incident/impinging angles at the scatterer(s). The delay $\tau_n = \frac{d_n}{c}$, where $c$ is the speed of light. 

In OFDM systems, the channel frequency response (CFR) of the $l^{th}$ sub-carrier can be expressed as 
\begin{equation}
\begin{split}
H(f,l)=\sum_{n=0}^{N-1}\alpha_{n}e^{-j2\pi f\tau_n+j\phi_n}e^{-j2\pi\tau_{n}f_{l}},
\end{split}
\label{CFR}
\end{equation}
where $f_{l}$ is the frequency of the $l^{th}$ subcarrier relative to the center frequency $f$. Now, we define the $1\times L$ channel vector ${\mathbf{H}(f)}=\{H(f,0),...,H(f,L-1)\} $ as the CFR of frequency $f$, and $\mathbf{H}_{1}=\mathbf{H}(f_{BA})$ as the CFR of Band1, $\mathbf{H}_{2}=\mathbf{H}(f_{AB})$ as the CFR of Band2, where $L$ is the total number of sub-carrier.

\subsection{Attack Model}
Based on the assumptions of most key generation schemes \cite{8314118}, \cite{9123376}, we also focus on passive adversary. The eavesdropper, Eve, is assumed to be located at least half of the larger wavelength in Band1 and Band2 away from the legitimate users, which can be mathematically given as
\begin{equation}
\begin{split}
d>\max\{\frac{c}{2f_{AB}}, \frac{c}{2f_{BA}}\}.
\end{split}
\label{d}
\end{equation}
Since wireless channel gains decorrelate over half a wavelength in a multipath environment, Eve's channel is assumed to be independent of the channel of legitimate users.  Therefore, Eve cannot infer the channel between legitimate users to generate the key from the channels he listens on. 


\subsection{Key Generation Scheme for FDD Systems}
According to the channel model, the carrier frequency will affect the phase and amplitude of the CFR, and the CFR difference is more obvious in the superposition of multiple paths at different frequencies \cite{zhang2016efficient}. 
This paper aims to construct reciprocal channel features in FDD systems for Alice and Bob based on the feature mapping function between different frequency bands. In Section~\ref{DL}, we first proved that there is a feature mapping function between different frequency bands that can be obtained through deep learning. We designed a KGNet-based key generation scheme, as shown in Fig. \ref{KGNet_ALL_1}. It consists of two stages, i.e., the KGNet training and KGNet-based key generation stages. 
\begin{itemize}
	\item We designed a KGNet model to learn the feature mapping function, which will be described in detail in Section \ref{KGNet}. Specifically, Alice trains KGNet to learn the feature mapping function between Band1 CFR $\mathbf{H}_{1}$ and Band2 CFR $\mathbf{H}_{2}$ stored in the database. The dataset in the database can be obtained by Alice collecting the CSI obtained by multiple channel detection measurements and the CSI feedback from Bob. The trained KGNet model will be used for key generation.
	\item In the KGNet-based key generation stage,  Alice and Bob send a pilot signal at the same time and perform channel estimation to obtain $\mathbf{H}_{1}$ and $\mathbf{H}_{2}$, respectively. Alice and Bob then preprocess their channel vector and obtain channel features $\mathbf{x}_{1}$ and $\mathbf{x}_{2}$. Alice will use KGNet to map the Band1 features $\mathbf{x}_{1}$ to the estimated Band2 features $\widehat{\mathbf{x}}_{2}$, which enables Alice and Bob to obtain highly correlated channel characteristics $\widehat{\mathbf{x}}_{2}$ and $\mathbf{x}_{2}$, respectively.
They will finally perform quantization, information reconciliation, and privacy amplification to generate the same key. The key generation scheme will be described in Section \ref{KGNet-based KG}. 
\end{itemize}

\begin{figure*}[!t]
	\centering
	\includegraphics[width=\linewidth]{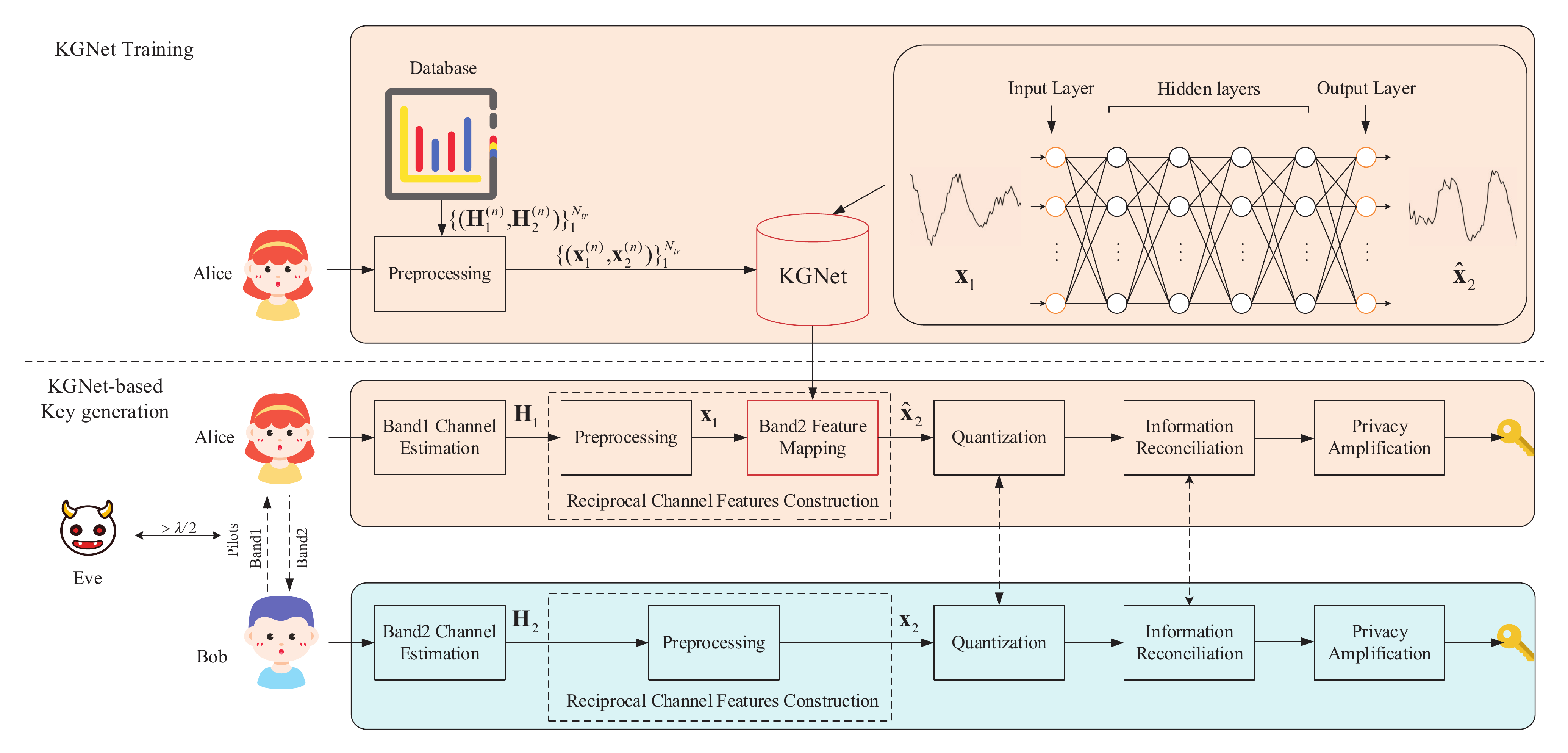}
	\caption{The KGNet-based Key Generation scheme for FDD systems. }
	\label{KGNet_ALL_1}
\end{figure*}

\section{Deep Learning for Band Feature Mapping} 
\label{DL}

In this paper, our main theme is to construct reciprocal channel features for the FDD system to generate a key according to the feature mapping function between frequency bands. However, according to (\ref{CFR}), we CANNOT determine whether the features between different frequency bands have a definite mapping function. And if there is a mapping function, how can we map the features on one frequency band to the features on another frequency band?

In this section, we define the band feature mapping function according to the approach in \cite{alrabeiah2019deep}. Then, we prove that leveraging deep learning can find the mapping function.

\subsection{Existence of the Band Mapping}
Consider the channel model in (\ref{CFR}), the channel is completely defined by the parameters $\alpha_{n}$, $\tau_n$, $\phi_n$, which are functions of the environment and the carrier frequency. Supposed that Alice is fixed, for a given static communication environment, there exists a deterministic mapping function from the position $P_B$ of Bob to the channel $\mathbf{H}{(f)}$ at every antenna element $m$ of Alice \cite{vieira2017deep}. Therefore, when the number of antenna is 1, there is also a deterministic position-to-channel mapping function when Alice and Bob are both equipped with a single antenna.
\begin{myDef}
	The position-to-channel mapping $\boldsymbol{\Phi}_f$ can be written as follows,
	\begin{equation}
	\begin{split}
	\boldsymbol{\Phi}_f:\{P_B\}\rightarrow\{\mathbf{H}(f)\},
	\end{split}
	\end{equation}
where the sets $\{P_B\}$ represent the possible positions of Bob and the sets $\{\mathbf{H}(f)\}$ represent the CFR of the corresponding channels. 
	\label{definition1}
\end{myDef}

Then, we investigate the existence of the mapping from the channel vector $\mathbf{H}(f)$ to the position of Bob. For that, we adopt the following assumption.
\begin{assumption}
	The position-to-channel mapping $\boldsymbol{\Phi}_f:\{P_B\}\rightarrow\{\mathbf{H}(f)\}$, is bijective.
	\label{assumption1}
\end{assumption}

This assumption means that every position of Bob has a unique channel vector $\mathbf{H}(f)$. The bijectiveness of this mapping depends on the number of subcarries, the set of candidate user locations and the surrounding environment. In massive MIMO systems, the probability that $\boldsymbol{\Phi}_f$ is bijective is actually very high in practical wireless communication scenarios, and approaches 1 as the number of antennas at the BS increased \cite{alrabeiah2019deep}. The same with a OFDM-FDD system, the probability that this mapping is bijective is also very high in practical wireless communication scenarios, and approaches 1 as the number of subcarries increases. Therefore, it is reasonable to adopt Assumption \ref{assumption1} in OFDM-FDD systems.

Under the Assumption \ref{assumption1}, we define the channel-to-position mapping $\boldsymbol{\Phi}_f^{-1}$ as the inverse of the mapping $\boldsymbol{\Phi}_f$, which can be written as
\begin{equation}
\begin{split}
\boldsymbol{\Phi}_f^{-1}:\{\mathbf{H}(f)\} \rightarrow \{P_B\}.
\end{split}
\end{equation}

\begin{proposition}
	If Assumption \ref{assumption1} is satisfied, there exists a channel-to-channel mapping function for a given communication environment, which can be written as follows,
	\begin{equation}
	\begin{split}
	\boldsymbol{\Psi}_{f_{BA}\rightarrow f_{AB}} = \boldsymbol{\Phi}_{f_{AB}} \circ \boldsymbol{\Phi}_{f_{BA}}^{-1}:\{\mathbf{H}(f_{BA})\} \rightarrow \{\mathbf{H}(f_{AB})\},
	\end{split}
	\end{equation}
where $\boldsymbol{\Phi}_{f_{AB}} \circ \boldsymbol{\Phi}_{f_{BA}}^{-1}$ represents the composite mapping related to $\boldsymbol{\Phi}_{f_{AB}}$ and $\boldsymbol{\Phi}_{f_{BA}}^{-1}$.
	\label{proposition1}
\end{proposition}

\begin{proof}
	See Appendix \ref{proof1}.
\end{proof}

Proposition \ref{proposition1} illustrates that channels on different frequency bands have a certain mapping function. Therefore, we believe that it is possible to construct reciprocal channel features from different frequency bands to generate keys.

\subsection{Deep Learning for Band Feature Mapping}
Since the channel vectors are all complex numbers and differ in the order of magnitude of each subcarrier, they cannot be directly used in the deep learning algorithm and the quantization part of subsequent key generation. Therefore, we propose a feature extraction mapping function $\boldsymbol{\xi}$ to preprocess $\mathbf{h}$, which can be written as 
\begin{equation}
\begin{split}
\boldsymbol{\xi}_f:\mathbf{H}(f)\rightarrow\mathbf{x}(f).
\end{split}
\end{equation}

Supposed $\boldsymbol{\xi}_f$ is linear, we can denote the inverse mapping of $\boldsymbol{\xi}$ as $\boldsymbol{\xi}^{-1}$ that can be given as
\begin{equation}
\begin{split}
\boldsymbol{\xi}_f^{-1}:\mathbf{x}(f) \rightarrow \mathbf{H}(f).
\end{split}
\label{inverse of feature extraction mapping}
\end{equation}

Next, we investigate the existence of the Band feature mapping as shown below.

\begin{proposition}
	With Proposition \ref{proposition1}, there exists a Band feature mapping function for a given communication environment, which can be written as follows,
	\begin{equation}
	\begin{split}
	\boldsymbol{\Psi}^{'}_{f_{BA}\rightarrow f_{AB}} =\boldsymbol{\xi}_{f_{AB}} \circ \boldsymbol{\Psi}_{f_{BA}\rightarrow f_{AB}} \circ \boldsymbol{\xi}_{f_{AB}}^{-1} : \mathbf{x}(f_{BA}) \rightarrow \mathbf{x}(f_{AB}).
	\end{split}
	\end{equation}
\label{proposition2}
\end{proposition}
\begin{proof}
	See Appendix \ref{proof2}.
\end{proof}

Although Proposition \ref{proposition2} proves that there is a feature mapping function between frequency bands, this function cannot be expressed as mathematical formulas. Therefore, based on the universe mapping \cite{hornik1989multilayer}, we introduce deep learning and obtain the following theorem.
\begin{theorem}
	For any given error $\varepsilon>0$, there exists a positive constant $M$ large enough such that
	\begin{equation}
	\begin{split}
	\sup_{\boldsymbol{x}\in\mathbb{H}} \parallel \textbf{NET}_M(\mathbf{x}(f_{BA}),\boldsymbol{\Omega})-\boldsymbol{\Psi}^{'}_{f_{BA}\rightarrow f_{AB}}(\mathbf{x}(f_{BA})) \parallel \leq \varepsilon,\\
	\mathbb{H}=\{\mathbf{x}(f_{BA})\},
	\end{split}
	\end{equation}
where $\textbf{NET}_M(\mathbf{x}(f_{BA}))$ is the output of a feedforward neural network with only one hidden layer. $\mathbf{x}(f_{BA})$, $\boldsymbol{\Omega}$ and $M$ denote the input data, network parameters, and the number of hidden units, respectively.
\label{theorem1}
\end{theorem}

\begin{proof}
	See Appendix \ref{proof3}.
\end{proof}

Theorem \ref{theorem1} reveals that the Band feature mapping function can be approximated arbitrarily well by a feedforward network with a single hidden layer. Thus, we can use deep learning to obtain the feature mapping function between frequency bands, and generate reciprocal channel features to generate the same key.

\section{KGNet-based Reciprocal Channel Features Construction}
\label{KGNet}


As a feedforward network with a single hidden layer can approximate the mapping of band features with different carrier frequencies, we propose a KGNet for channel feature mapping to construct reciprocal channel features. We will first introduce the dataset preprocessing. Then, we introduce the KGNet architecture and describe how to train and test the KGNet.

\re{
\subsection{Dataset Generation}
Since the wireless environment is complex and changeable, there are datasets $\{\mathbb{D}_{O}^{e}\}_{e=1}^E$ of E different environments, and $E \to \infty$. It is impossible to get datasets in all different environments. This paper currently only considers a scenario in a given environment. In the future, techniques such as transfer learning and meta learning can also be used to extend this method to new environments. 

In a given environment, we denote the CSI of Band1 and Band2 as the original dataset, which is divided into the training dataset and testing dataset as $\mathbb{D}_{OTr}$ and $\mathbb{D}_{OTe}$, respectively. $\mathbb{D}_{OTr}=\{(\mathbf{H}_{1}^{(n)},\mathbf{H}_{2}^{(n)})\}_{n=1}^{N_{tr}}$, including $N_{tr}$ training label samples. $\mathbb{D}_{OTe}=\{(\mathbf{H}_{1}^{(n)},\mathbf{H}_{2}^{(n)})\}_{n=1}^{N_{te}}$, including $N_{te}$ testing label samples. It should be emphasized that the training dataset cannot include all possible channels between Alice and Bob. As long as the training dataset is sufficiently representative of the environment, the trained neural network can learn the band feature mapping of the environment. Therefore, Alice only needs to collect enough data that is sufficient to represent the environment.
}

\subsection{Preprocessing of Dataset}
\label{Preprocessing of Dataset}
In order to enable the deep learning model to perform efficiently, data samples usually go through a sequence of preprocessing operations, including realization and normalization. The first operation is realization. Since deep learning algorithms work in real domain, we introduce the mapping $\boldsymbol{\xi}^{(1)}$ to stack the real and imaginary parts of the complex channel vector, which can be written as 
\begin{equation}
\begin{split}
\boldsymbol{\xi}^{(1)} :\mathbf{H}'\rightarrow(\mathfrak{R}(\mathbf{H}),\mathfrak{T}(\mathbf{H})),
\end{split}
\end{equation}
where $\mathfrak{R}(\cdot)$ and $\mathfrak{T}(\cdot)$ denote the real and imaginary parts of a matrix, vectors or scales, respectively. Through realization, the complex channel vector becomes $1\times2L$ real channel vector $\mathbf{H}'_1$ and $\mathbf{H}'_2$.

The second operation is normalization. Since each dimension of the original dataset usually has a different order of magnitude, directly using the original dataset to train the network will affect the efficiency of the network. Normalization is commonly used to normalize the dataset so that the range of the dataset is between 0 and 1. It is performed using the maximum and minimum value in the dataset. The value is given by:
\begin{equation}
\begin{split}
\begin{cases}
H_{1,max}^{'l}=\max\limits_{n=1,...,N_{tr}}\{H_1^{'l}\}^{n}\\
H_{1,min}^{'l}=\min\limits_{n=1,...,N_{tr}}\{H_1^{'l}\}^{n}
\end{cases}
l=0,...,2L-1,
\end{split}
\end{equation}
where $H_1^{'l}$ is the $l$th element of $\mathbf{H}'_1$. We introduce the mapping $\boldsymbol{\xi}^{(2)}$ to normalize the dataset, which can be written as:
\begin{equation}
\begin{split}
\boldsymbol{\xi}^{(2)}:
x_{1}^{l} \rightarrow \frac{H_{1}^{'l}-H_{1,min}^{'l}}{H_{1,max}^{'l}-H_{1,min}^{'l}},l=0,...,2L-1  ,
\end{split}
\end{equation}
where $x_1^{l}$ is the $l$th element of $\mathbf{x}_1$. The normalization of $\mathbf{H}'_2$ follows the same procedure.

After the above processing, the training dataset and testing dataset as $\mathbb{D}_{Tr}=\{(\mathbf{x}_{1}^{(n)},\mathbf{x}_{2}^{(n)})\}_{n=1}^{N_{tr}}$ and $\mathbf{D}_{Te}=\{(\mathbf{x}_{1}^{(n)},\mathbf{x}_{2}^{(n)})\}_{n=1}^{N_{te}}$ that can be directly put into the network are obtained. We can define $\boldsymbol{\xi}_f=\boldsymbol{\xi}^{(2)}\circ\boldsymbol{\xi}^{(1)}$, which is a linear transformation that satisfies the supposition of (\ref{inverse of feature extraction mapping}). 

Note that $H_{1,max}^{'l}$, $H_{1,min}^{'l}$, $H_{2,max}^{'l}$, $H_{2,min}^{'l}$ used to normalize the testing dataset are all derived from training dataset.

\subsection{KGNet Architecture}
\label{KGNet Architecture}
Based on the feedforward network, we propose the KGNet, which consists of  one input-layer, four hidden-layer and one output-layer. As shown in Fig. 1, the input of the network is $\mathbf{x}_{1}$ obtained after preprocessing $\mathbf{H}_{1}$.  The output of the network is a cascade of nonlinear transformation of $\mathbf{x}_{1}$, i.e.,
\begin{equation}
\begin{split}
\widehat{\mathbf{x}}_{2}=\textbf{KGNet}(\mathbf{x}_{1},\boldsymbol{\Omega}),
\end{split}
\end{equation}
where $\Omega$ is all the trainable parameters in this network. Obviously, KGNet is used to deal with a vector regression problem. The activation functions for the hidden layers and the output layer are the rectified linear unit (ReLU) function and the sigmoid function, respectively. 

\subsection{Training and Testing}
\label{Training and Testing}

In the training stage, the dateset $\mathbb{D}_{Tr}$ is collected as the complete training dataset. In each time step, $V$ training samples are randomly selected from $\mathbb{D}_{Tr}$ as $\mathbb{D}_{TrB}$. 
The KGNet is trained to minimize the difference between the output $\widehat{\mathbf{x}}_{2}$ and the label $\mathbf{x}_2$ by the adaptive moment estimation (ADAM) algorithm \cite{kingma2014adam}. The loss function can be written as follow: 
\begin{equation}
\begin{split}
Loss_{\mathbb{D}}(\boldsymbol{\Omega})=MSE(\widehat{\mathbf{x}}_{2}, \mathbf{x}_{2})=\frac{1}{VN_\mathbf{x}}\sum_{v=0}^{V-1}\|\widehat{\mathbf{x}}_{2}^{(v)}-\mathbf{x}_{2}^{(v)}\|_2^2,
\end{split}
\end{equation}
where $V$ is the batch size, the superscript $(v)$ denotes the index of the $v$-th training sample, $N_\mathbf{x}$ is the length of the vector $\mathbf{x}_{2}$, $\mathbb{D}=\{(\mathbf{x}_{1},\mathbf{x}_{2})\}_{v=0}^{V-1}$ is a batch-sized training dataset, and $\|\cdot\|_2$ denotes the $\ell_2$ norm. 
Various loss functions can be used to train the neural network, e.g., mean square error (MSE), mean absolute error (MAE), and logcosh. Since the performance of using these loss functions is basically the same when dealing with this problem, we choose MSE as the loss function, which is also used in most works to deal with similar problems \cite{yang2019deep}.

In the testing stage, the KGNet parameter $\boldsymbol{\Omega}$ is fixed. The testing dataset $\mathbb{D}_{Te}$ is generated and is used to test the performance of the KGNet. Furthermore, the dataset can also be used to evaluate the performance of the initial key after quantization.
 

\section{The KGNet-based Key Generation Scheme}
\label{KGNet-based KG}
The KGNet-based key generation scheme for FDD systems consists of five steps including channel estimation, reciprocal channel features construction, quantization, information reconciliation and privacy amplification, which are designed in this section.

\subsection{Channel Estimation}
\label{Channel Estimation}
Channel estimation is the first step for the key generation scheme, during which Alice and Bob send a pilot signal to each other simultaneously and estimate the CSI. Through this step, Alice and Bob can obtain the channel coefficient vector of Band1 and Band2, namely $\mathbf{H}_1$ and $\mathbf{H}_2$, respectively. 

\subsection{Reciprocal Channel Features Construction}
Reciprocal channel features construction is a process in which Alice and Bob construct reciprocal channel features according to the CSIs estimated in Section \ref{Channel Estimation}, which is the most important step for the key generation scheme. 

For Alice, this step consists of two parts: preprocessing channel coefficient vector $\mathbf{H}_1$ and Band2 feature mapping. The first part includes realization and normalization, as explained in Section \ref{Preprocessing of Dataset}. In the second part, Alice performs feature mapping using the pre-trained KGNet designed in Section \ref{Training and Testing}. For Bob, this step only involves the preprocessing channel coefficient vector $\mathbf{H}_2$. Through the above steps, Alice and Bob can obtain the $1\times2L$ channel feature vector $\widehat{\mathbf{x}}_2$ and $\mathbf{x}_2$, respectively.

\subsection{Quantization}
Upon acquiring the channel features, Alice and Bob should apply the same quantization algorithm to convert channel features into a binary bit stream with low key error rate (KER). In this paper, based on the equal probability quantization method, we propose a Gaussian distribution-based quantization method with guard-band (GDQG) for OFDM systems. 

Different from \cite{zenger2015security} that estimates the mean and variance of the amplitudes of CFR samples in a time sequence, we estimate the mean and variance of the real and imaginary parts of CFR across subcarriers, which can be expressed as
\begin{equation}
	\begin{split}
		\mu=\frac{1}{2L}\sum_{l=0}^{2L-1}x^l,
	\end{split}
\end{equation}
\begin{equation}
	\begin{split}
		\sigma^2=\frac{1}{2L-1}\sum_{l=0}^{2L-1}(x^l-\mu)^2,
	\end{split}
\end{equation}
The distribution of channel features can be approximated by a Gaussian distribution. Therefore, we fit the probability of the channel features into a definite Gaussian distribution $\mathcal{N}_Q=\mathcal{N}(\mu,\sigma^2)$.

The $k^{th}$ quantization interval is calculated as,
\begin{equation}
	[F^{-1}(\frac{k-1}{K}+\varepsilon), F^{-1}(\frac{k}{K}-\varepsilon)], k=2,...,K-1
	  ,
	  \label{equation20}
\end{equation}
and the $1^{st}$ quantization interval is $[0, F^{-1}(\frac{1}{K}-\varepsilon)]$, $K^{th}$ quantization interval is $[F^{-1}(\frac{K-1}{K}+\varepsilon,1]$, where $F^{-1}$ is defined as the inverse of the cumulative distribution function (CDF) of $\mathcal{N}_Q$ and $K$ is the quantization level. The $\varepsilon \in (0,1/2K)$ is defined as the quantization factor, which is used to set the limit of the guard band. Then, we use the common binary encoding to convert the features into a bit stream, where all the features that are not in the quantization intervals are set to -1.



The steps of the quantization method are given in the Algorithm \ref{alg:quantization}.
\begin{algorithm}[htb] 
	\caption{Gaussian distribution-based quantization method with guard-band} 
	\label{alg:quantization} 
	\begin{algorithmic}[1] 
		\REQUIRE 
		The band feature vector $\mathbf{x}$;  The quantization factor $\varepsilon$;
		\ENSURE 
		The quantized binary sequence $\mathbf{Q}$;	
		\STATE Calculate the mean $\mu$ and variance $\sigma^2$ of the feature vector $\mathbf{x}$;
		\STATE Construct Gaussian distribution $\mathcal{N}_Q=\mathcal{N}(\mu,\sigma^2)$;
		\STATE Calculate the inverse of the CDF $F^{-1}$ of $\mathcal{N}_Q$;
		\STATE Initialize the value of $\varepsilon$;
		\STATE Calculate the quantization intervals using (\ref{equation20});
		\FOR{$i = 0 : 2L-1$}
		\IF {$x^{i} \in k^{th}$ quantization interval}
		\STATE Binary encoding
		\ELSE
		\STATE Coded as -1
		\ENDIF
		\ENDFOR
		\RETURN $\mathbf{Q}$.
	\end{algorithmic}
\end{algorithm}
Finally, Alice and Bob send each other indexes whose values are -1 and delete all these bits.
They can get the initial secret key $\mathbf{Q}_A$ and $\mathbf{Q}_B$.


\subsection{Information Reconciliation and Privacy Amplification}
To further reduce the KER, we can adopt information reconciliation to correct the mismatch bits. Through information reconciliation techniques, which can be implemented with protocols such as Cascade \cite{zhu2013extracting} or BCH code \cite{dodis2004fuzzy}, etc, the performance of KER can be significantly improved. 
Privacy amplification can apply hash functions to distill a shorter but secret key. However, to ensure fairness of comparison, we only compare the performance of the initial secret key without information reconciliation and privacy amplification.

\section{Simulation Results}
\label{Simulation}
In this section, we evaluate the performance of our proposed deep learning-based key generation scheme for FDD systems. We first describe the simulation setup and metrics, and then we discuss the simulation results and overhead.

\subsection{Simulation Setup}
\label{Simulation Setup}
\begin{figure}[!t]
	\centering \includegraphics[width=\linewidth]{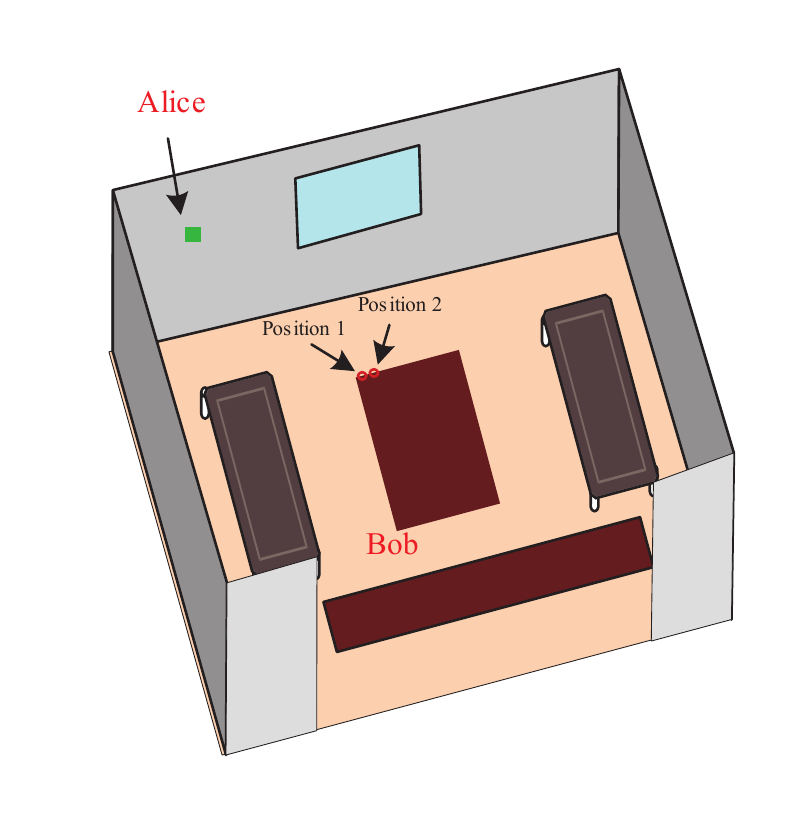} 
	\caption{An approximate depiction of the considered environment. The green little box on the ceiling represents the distributed antennas of the base station. The two maroon rectangles are two grids representing possible user locations.}  
	\label{scenario} 
\end{figure}

We consider the indoor distributed massive MIMO scenario `I1' that is offered by the DeepMIMO dataset \cite{alkhateeb2019deepmimo} and is generated based on the accurate 3D ray-tracing simulator Wireless InSite \cite{Remcom}. 
As depicted in Fig. \ref{scenario}, the model involves one BS with single antenna and 100,000 users. We assume that the BS is Alice and multiple users are possible locations for Bob and Eve. This scenario is available at two operating frequencies 2.4 GHz and 2.5 GHz, which emulate the Band1 and Band2 carrier frequencies of the singer-input singer-output (SISO) setup in Section \ref{System Model}. We set the number of OFDM subcarriers as 64, the number of paths as 5, and the bandwidth as 0.5 GHz to generate the dataset. This dataset constructs the channels between every candidate user location and a single antenna at the Band1 and Band2 frequencies. To form the training and testing dataset, we take the first 100,000 of the generated dataset and shuffle, and then split into a training dataset $\mathbb{D}_{OTr}$ with 80\% of the total size and a testing dataset $\mathbb{D}_{OTe}$ with the rest 20\%. These datasets are used to train the KGNet and evaluate the performance of the proposed key generation scheme.


\begin{table}[!t] 
\caption{Parameters for the KGNet.} 
\label{tab:KGNet-setup}
	\centering
	\begin{tabular}{|p{5cm}<{\centering}| c|}
		\hline
		\bfseries Parameter & \bfseries Value\\ \hline
		Number of neurons in hidden layers& (512,1024,1024,512)\\\hline
		Optimization & ADAM \cite{kingma2014adam} \\ \hline
		Kernel initializer &  glorot\_uniform\\ \hline
		Bias initializer &  zeros\\ \hline
		Exponential decay rates for ADAM: ($\rho_1,\rho_2$) &(0.9,0.999)\\ \hline
		Disturbance factor for ADAM & 1e-8\\ \hline
		Learning rate& 1e-3\\\hline
		Number of epochs& 500\\\hline
		Batch size& 128\\\hline
		Number of training samples& 80,000\\\hline
		Number of testing samples& 20,000\\
		\hline
	\end{tabular}
\end{table}

The KGNet is implemented on a workstation with one Nvidia GeForce GTX 1660Ti GPU and one Inter(R) Core(TM) i7-9700 CPU. Tensorflow 2.1 is employed as the deep learning framework. The parameters of the KGNet are given in Table \ref{tab:KGNet-setup}.

\subsection{Performance Metrics}
\label{7.2}

We use the \textit{Normalized Mean Square Error (NMSE)} to evaluate the predictive accuracy of the network, which  is defined as
\begin{equation}
\begin{split}
\mathrm{NMSE}=E\left[\frac{\parallel\widehat{\mathbf{x}}_{2}-\mathbf{x}_2\parallel_2^2}{\parallel \mathbf{x}_{2}\parallel_2^2}\right],
\end{split}
\end{equation}
where $E\left[\cdot\right]$ represents the expectation operation. 

We evaluate the performance of the initial key using the following metrics.
\begin{itemize}
	\item \textit{Key Error Rate (KER)}: It is defined as the number of error bits divided by the number of total key bits.
	\item \textit{Key Generation Ratio (KGR)}: It is defined as the number of initial key bits divided by the number of subcarriers. If all the real and imaginary features of the subcarriers are used to generate the key bits and the guard band is not used during quantization, then the KGR reaches a maximum of 2.
	\item \textit{Randomness}: The randomness reveals the distribution of bit streams. The National Institute of Standards and Technology (NIST) statistical test \cite{rukhin2001statistical} will be used for the randomness test for the key.
\end{itemize}

We evaluate the performance at Eve by using the  \textit{Normalized Vector Distance (NVD)} between two vectors $\mathbf{K}_1$ and $\mathbf{K}_2$, which is defined as
\begin{equation}
\begin{split}
\mathrm{NVD(\mathbf{K}_1,\mathbf{K}_2)}=\frac{\parallel\mathbf{K}_1-\mathbf{K}_2\parallel_2^2}{\parallel \mathbf{K}_2\parallel_2^2}.
\end{split}
\end{equation}

\subsection{Results}
\label{7.3}
In this section, we evaluate the performance of KGNet and the initial key. Then, the security of the secret generation scheme proposed in this paper is analyzed.

\subsubsection{The Performance of KGNet}                         


The performance of the KGNet is critical to whether Alice and Bob can generate highly reciprocal channel characteristics. Besides KGNet, we also considered the following three benchmark models.
\begin{itemize}
	\item FNN. A FNN is originally designed in \cite{huang2019deep} for uplink/downlink channel calibration for massive MIMO systems, which can be used for the band feature mapping for SISO-OFDM system. It consists of three hidden layers and we choose the numbers of neurons in the hidden layer are (512, 1024, 512) by adjustments. FNN in \cite{huang2019deep} differs from KGNet proposed in this paper not only in the number of layers of the network, but also in all layers of the FNN, the tanh function is used as the activation function. 
	\item Advanced-FNN. It has the same architecture as FNN but its activation functions are the same as KGNet. 
	\item 1D-CNN. Its structure is to add a 1D convolutional layer with 8 convolution kernels of size 8 and a max pooling layer with step size 2 and pooling kernel size 2 between the input layer and the first hidden layer of Advanced-FNN. 
\end{itemize}
We will compare the KGNet with the three benchmark networks from three aspects, namely, the size of training data required, the fitting performance of the network, and the generalization performance under various SNR.

The size of training dataset is crucial in the KGNet training. Fig. \ref{Training_dataset_size} shows how the performance of the KGNet improves as the size of the training dataset grows. The network is trained from scratch for every training dataset size and is tested on the fixed 10,000 test set. We compare the size of training data required for each of the four networks to achieve optimal results. As expected, the performance of the four networks has improved as the size increased. However, all four networks achieve a stable NMSE after 80,000 training sets. Hence, we use 80,000 sets of data for training and 20,000 sets of data for testing in the rest of the paper. 
\begin{figure}[!t]
	\centering 
	\includegraphics[width=\linewidth]{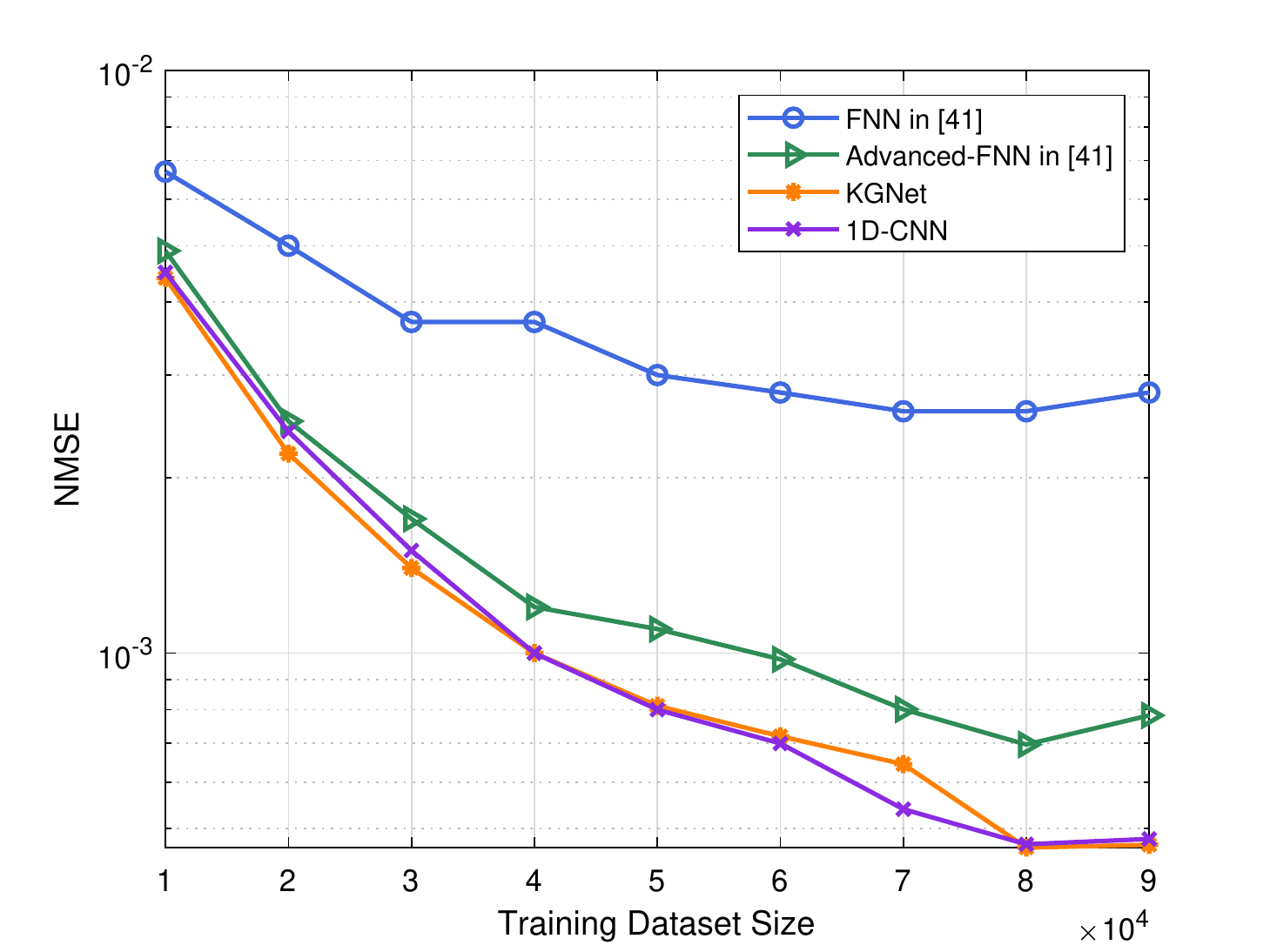} 
	\caption{The NMSE of the four networks versus the size of training dataset.}
	\label{Training_dataset_size} 
\end{figure}

In Fig. \ref{Performance_differentNN}, we compare their fitting performance. From the perspective of fitting performance, KGNet and 1D-CNN structures can be selected for band feature mapping. In addition, the NMSE of the four networks is at a basic level when the number of epochs is 500, which is used in the rest of the paper.
\begin{figure}[!t]
	\centering 
	\includegraphics[width=\linewidth]{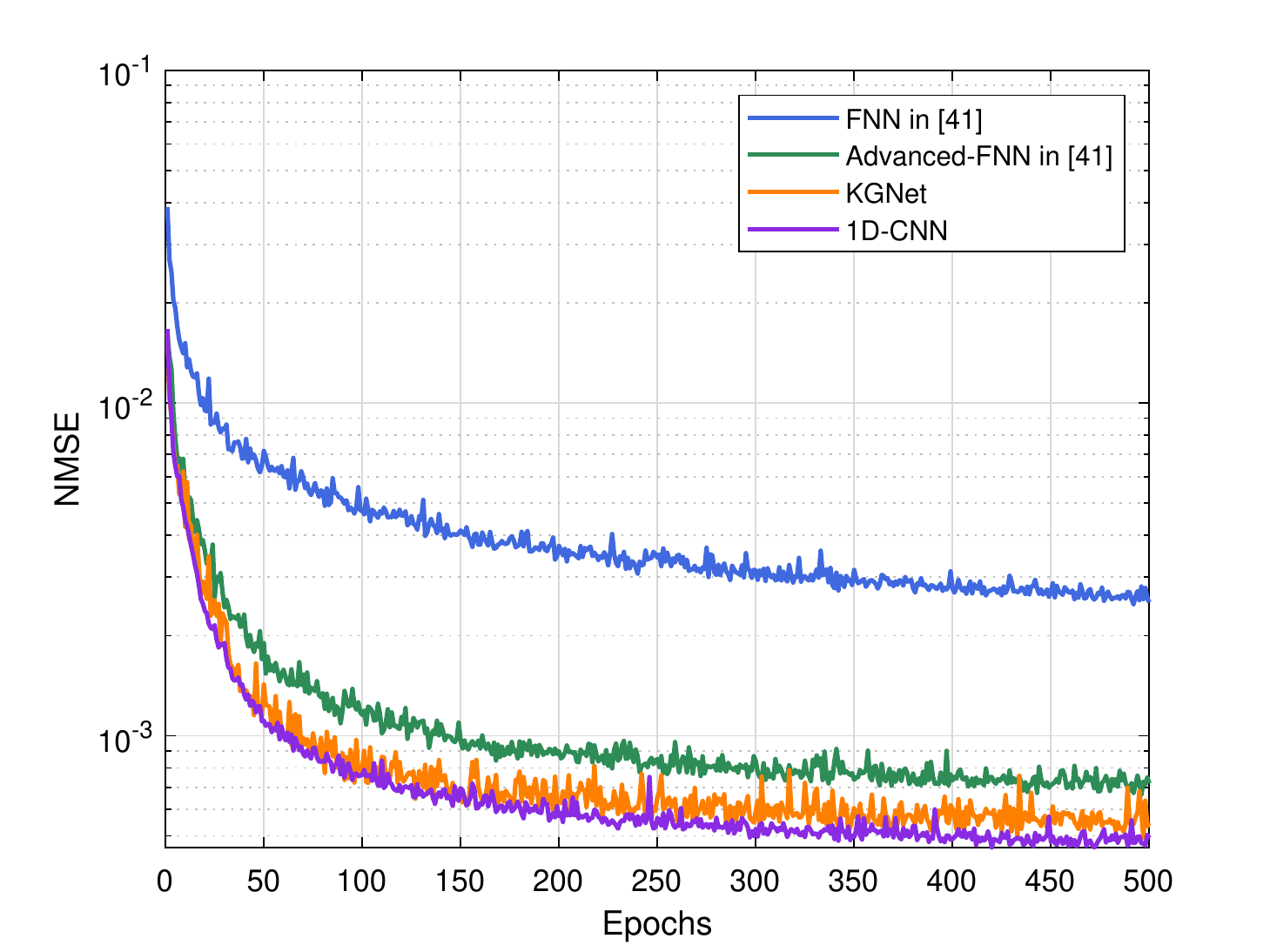} 
	\caption{The NMSE of the four networks versus the number of epochs.}
	\label{Performance_differentNN} 
\end{figure}

Moreover, the generalization of the network is critical to the performance of key generation. 
The four networks were trained with the original dataset (without noise).
We then added complex white Gaussian noise to the test dataset to generate new test datasets with different SNR levels in the range of 0-40 dB with a 5 dB step.
As shown in Fig. \ref{NMSE_SNR_differentNN}, as the SNR increases, the NMSE of the network decreased. Compared with the other three networks, KGNet has better performance in the range of 0-40 dB. 
\begin{figure}[!t]
	\centering 
	\includegraphics[width=\linewidth]{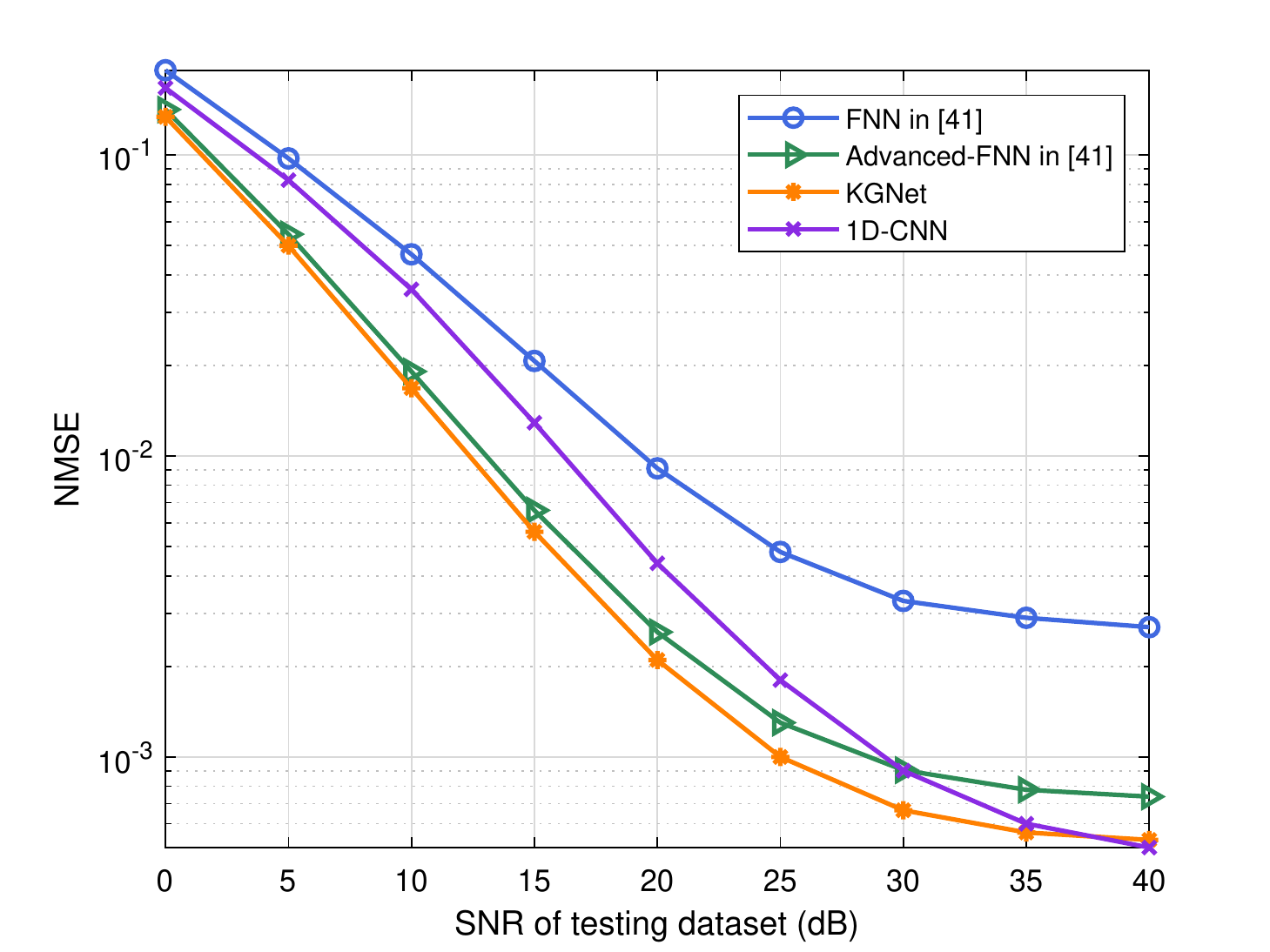} 
	\caption{The NMSE of the four networks versus SNR of test sets.}
	\label{NMSE_SNR_differentNN} 
\end{figure}

Based on the analysis in the above three aspects, KGNet is the best choice for band feature mapping. 
Besides the above model, we tried other architectures by adding more convolutional layers and full connection layers based on the 1D-CNN and Advanced-FNN. Although the performance of network is slightly improved under high SNR, it also declines under low SNR. Hence they are not compared in this paper.

\subsubsection{Adding Artificial Noise to Improve Robustness}
In order to improve the generalization performance of the network under low SNR, we added artificial noise to the original training dataset. Firstly, we used the dataset under a single SNR as the training dataset, and observed the NMSE of KGNet when tested under different SNRs. As shown in Fig. \ref{NMSE_SNR_differentSNR_training}, when the SNR of the training dataset is lower, the NMSE of KGNet is better when the SNR is lower than 15 dB, and the NMSE is worse when the SNR exceeds 15 dB. When the training dataset is the original noise-free data, the network has the best performance under high SNR while the network performs well under low SNR when the training dataset is SNR of 0 dB. Therefore, we choose to cross the original noise-free dataset with the dataset with SNR of 0 dB, so that the KGNet can achieve excellent performance under 0 - 40 dB.
\begin{figure}[!t]
	\centering 
	\includegraphics[width=\linewidth]{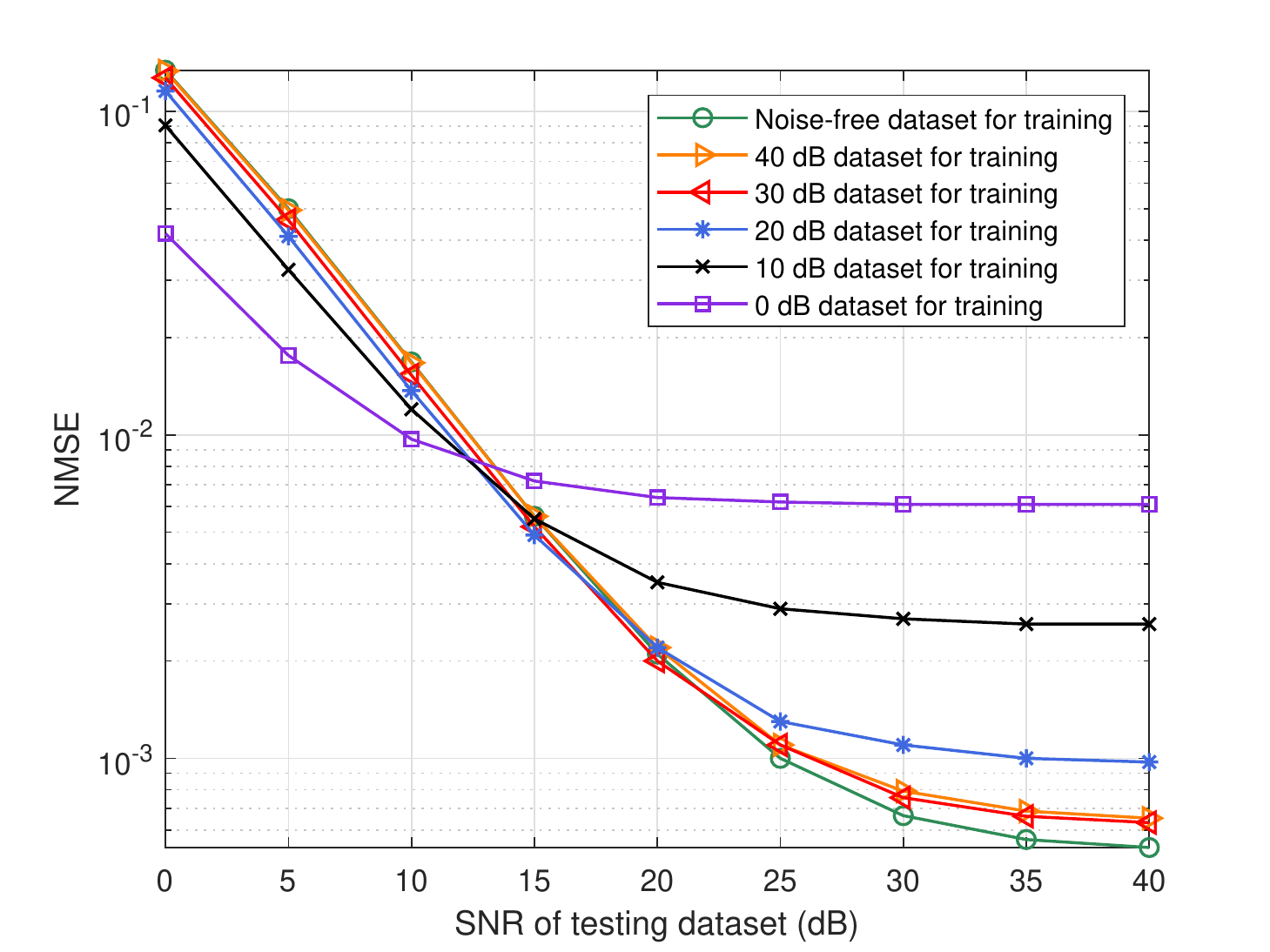} 
	\caption{The NMSE of the KGNet trained with different SNR dataset versus SNR of testing dataset.}
	\label{NMSE_SNR_differentSNR_training} 
\end{figure}

Then, we combined the noise-free dataset and the 0 dB dataset in different sizes to form the new training dataset. We constructed five cross datasets:
\begin{itemize}
	\item Cross dataset 1: A dataset combining 80,000 sets of 0 dB dataset and 0 sets of noise-free dataset; 
	\item Cross dataset 2: A dataset combining 60,000 sets of 0 dB dataset and 20,000 sets of noise-free dataset; 
	\item Cross dataset 3: A dataset combining 40,000 sets of 0 dB dataset and 40,0000 sets of noise-free dataset; 
	\item Cross dataset 4: A dataset combining 20,000 sets of 0 dB dataset and 60,000 sets of noise-free dataset; 
	\item Cross dataset 5: A dataset combining 0 sets of 0 dB dataset and 80,000 sets of noise-free dataset.
\end{itemize}
As shown in Fig. \ref{NMSE_SNR_crossSNR_training}, when the network is trained with the cross dataset 5, the performance of the network is the best when the SNR is above 25 dB, but its test performance under low SNR is too bad compared to the other four cross datasets. The test performance of the network trained with the cross datasets 1-4 is almost the same in the low SNR. Therefore, in order to ensure the excellent performance under high SNR and improve the performance under low SNR to a certain extent, we choose the cross dataset 4.
\begin{figure}[!t]
	\centering
	\includegraphics[width=\linewidth]{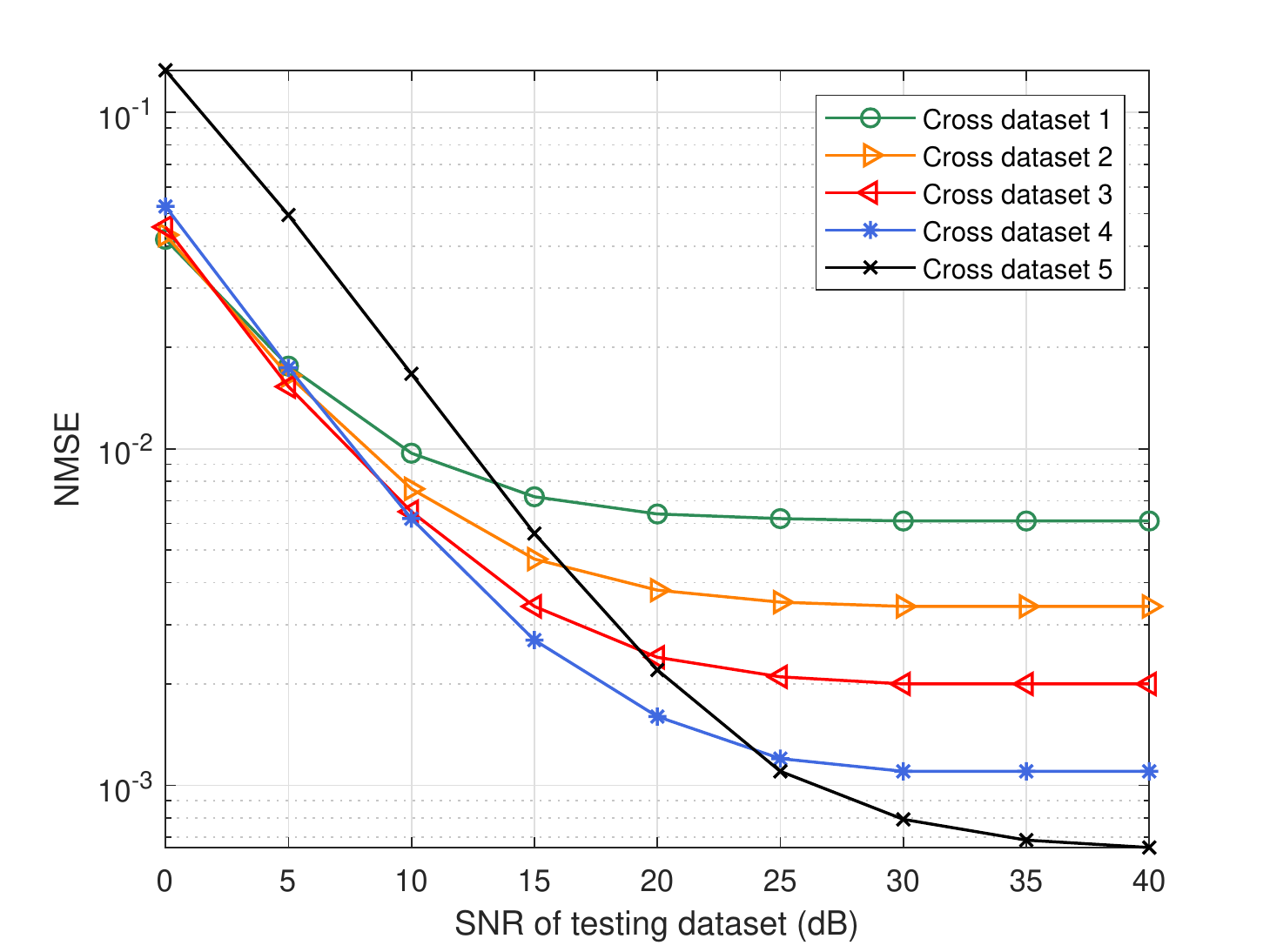} 
	\caption{The NMSE of the KGNet trained with cross datasets 1 - 5 versus SNR of testing dataset. }
	\label{NMSE_SNR_crossSNR_training} 
\end{figure}

Finally, we used the cross dataset 4 to train the four networks, and the NMSE of the four networks tested under different SNRs is shown in Fig. \ref{NMSE_SNR_differentNN_crossSNR_training}. Compared with the network performance shown in Fig. \ref{NMSE_SNR_differentNN}, the use of cross dataset to train the network reduces the performance of the four networks under high SNR, and the performance of the four networks under low SNR is improved. It is obvious that the generalization performance of the four networks has been improved, and the generalization performance of KGNet is still better than the other three networks.
\begin{figure}[!t]
	\centering 
	\includegraphics[width=\linewidth]{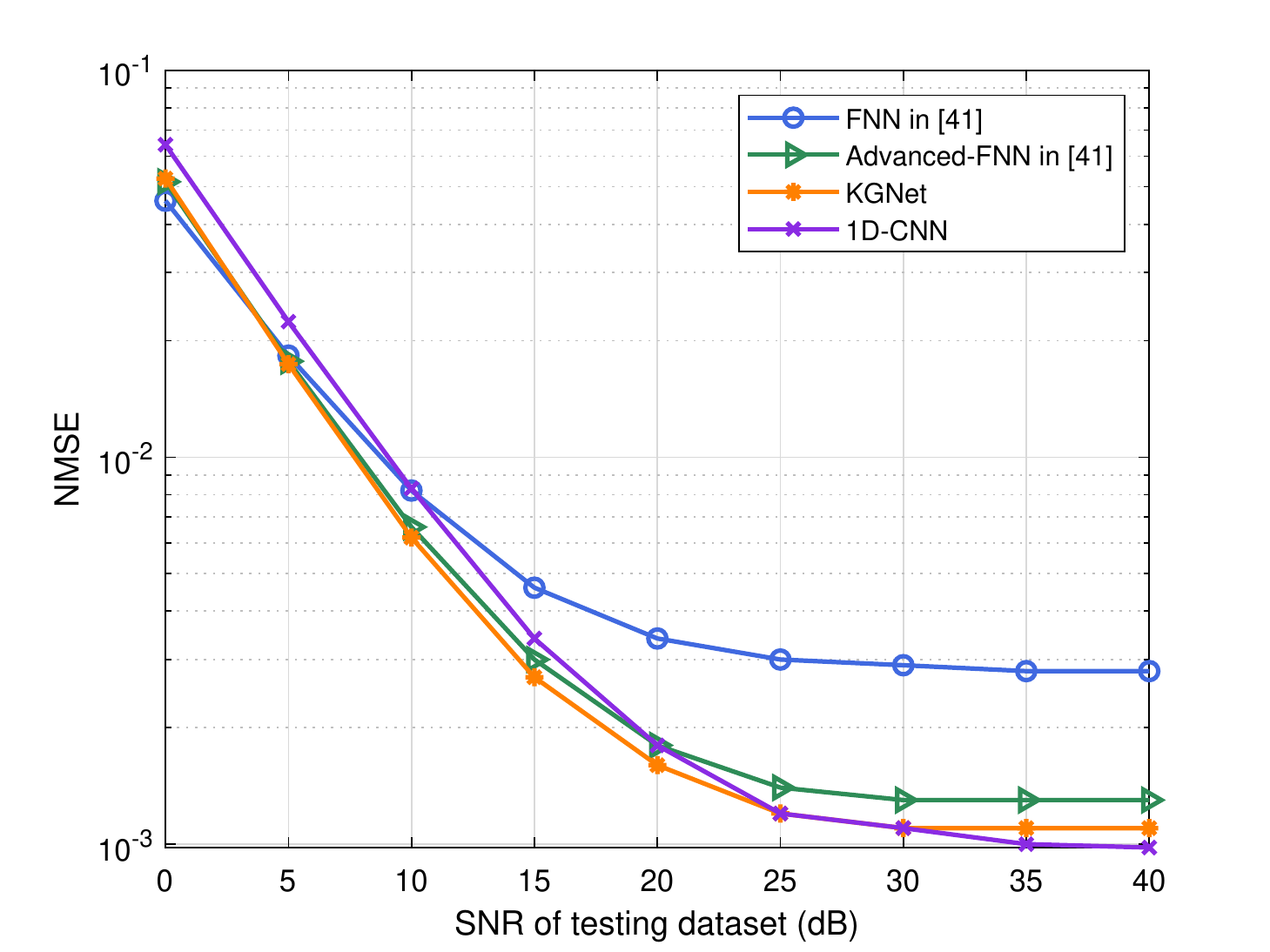} 
	\caption{The NMSE of the four networks versus SNR when all networks are trained with the cross dataset 4.}
	\label{NMSE_SNR_differentNN_crossSNR_training} 
\end{figure}

\subsubsection{The Performance of Initial Key}

The performance of the initial key plays the most important role in secret key generation problems as it provides the capability Alice and Bob can achieve the same secret keys. 

\begin{table}[!t] 
\caption{SNR thresholds to achieve KER = $10^{-1}$ when $\varepsilon=\frac{0.2}{2(K-1)}$.}
\label{tab:multiquantization}
	\centering
	\begin{tabular}{|p{3cm}<{\centering}| c|c|c|}
		\hline
		\bfseries  Quantization level& \bfseries $K=2^1$ & \bfseries $K=2^2$ & \bfseries $K=2^3$\\ \hline
		\textbf{SNR threshold (dB)}& 5 & 15 & --\\
		\hline
	\end{tabular}
\end{table}

In practice, the commonly used information reconciliation methods can correct the initial key with KER less than $10^{-1}$. Table \ref{tab:multiquantization} compares the lowest SNRs under different quantization levels to achieve this goal when $\varepsilon=\frac{0.2}{2(K-1)}$. $\varepsilon=\frac{0.2}{2(K-1)}$ means that regardless of the level of quantization, about 20\% of the channel characteristics are discarded. When $K=2^{3}$, the KER under 40 dB is 0.15, which still cannot meet the target. In order to better analyze the impact of other parameters on the performance of the initial key, the quantization level in the subsequent analysis is $2^{1}$.

Fig. \ref{KERKGR_SNR_differentNN} compares the average KER performance of the four networks under 20,000 testing dataset. We chose the quantization factor $\varepsilon$ of 0.1.
As the SNR increases, the KER of the four networks decreases, and the KER of KGNet is always lower than the other networks. When the SNR is higher than 30 dB, the KER performance of KGNet is lower than $10^{-3}$. Furthermore, Fig. \ref{KERKGR_SNR_differentNN} compares the average KGR performance of the three networks under 20,000 testing dataset. When the SNR is low and the network performance is not good enough, there will be more features in the isolation band, and its KGR performance is reduced. At the theoretical level, we choose the quantization factor to be  0.1, which means that about 20\% of the channel features in the total $2L$ channel features are deleted during the quantization process, so the number of quantized key bits generated is $1.6L$ and the KGR is 1.6. As shown in Fig. \ref{KERKGR_SNR_differentNN}, with the growth of SNR, the KGR performance of the four networks also continues to increase and is close to 1.6 when it is above 25 dB, which means the Gaussian distribution fits the distribution of channel features well.

\begin{figure}[!t]
	\centering 
	\includegraphics[width=\linewidth]{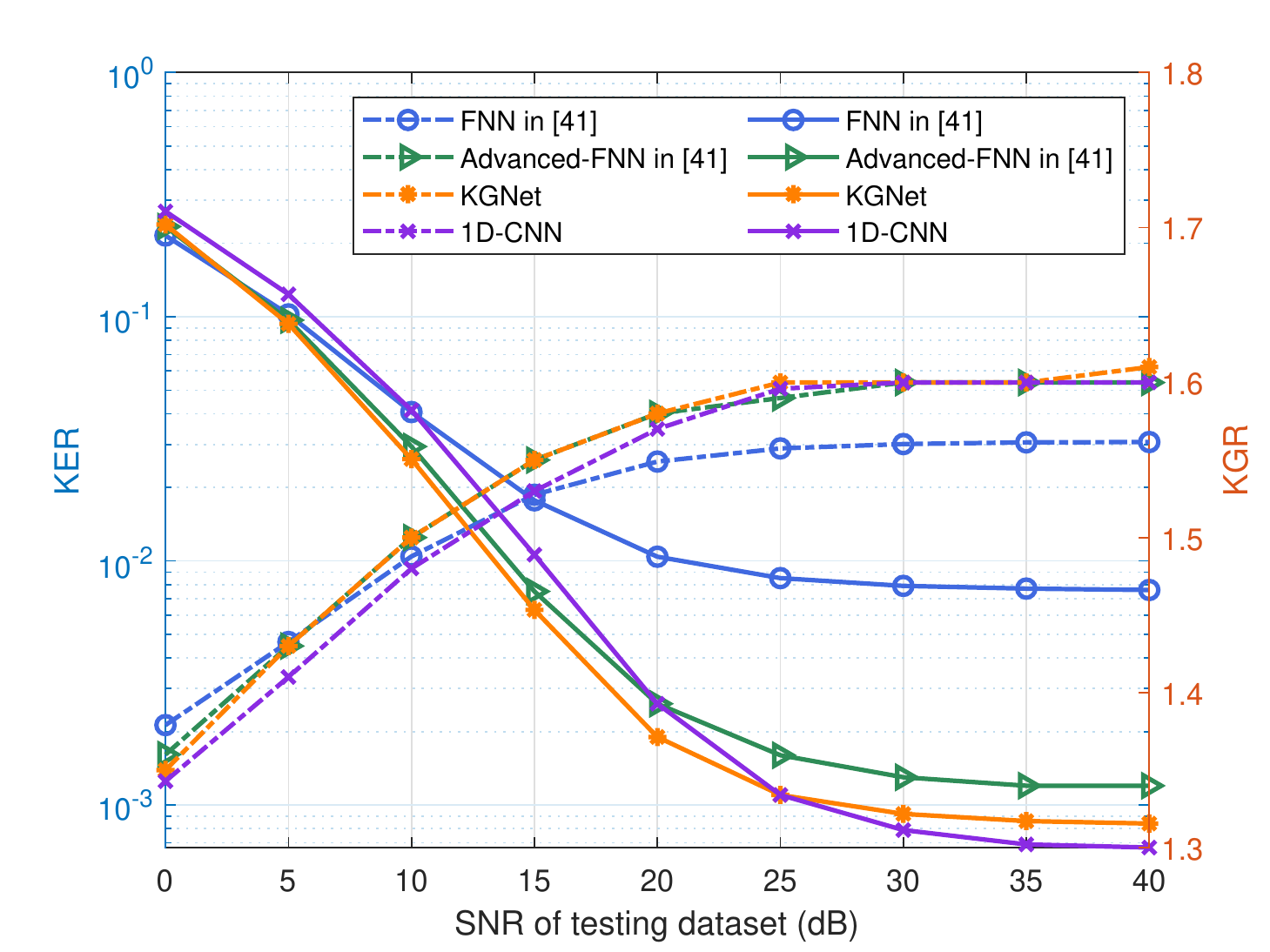} 
	\caption{The KER and KGR of the initial key based on the four networks versus SNR. All networks are trained with the cross dataset 4. The quantization factor is 0.1. The solid line represents KER and the dashed line represents KGR.}
	\label{KERKGR_SNR_differentNN} 
\end{figure}
\begin{figure}[!t]
	\centering 
	\includegraphics[width=\linewidth]{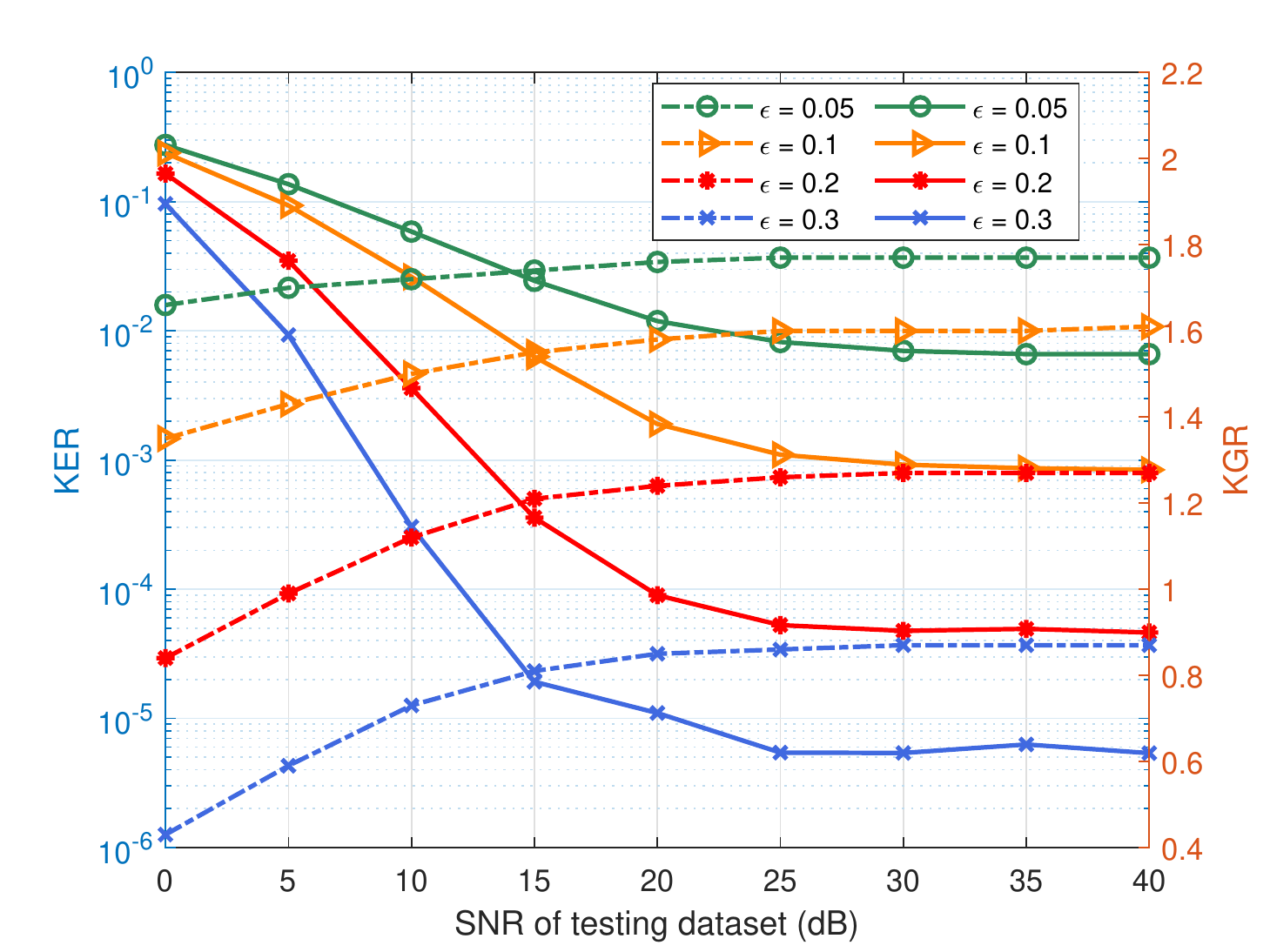} 
	\caption{The KER and KGR of the initial key based on the KGNet versus SNR. The KGNet is trained with the cross dataset 4.  The solid line represents KER and the dashed line represents KGR. }
	\label{KERKGR_SNR_quantization_factor} 
\end{figure}

%
Additionally, the choice of the quantization factor has a great influence on the performance of KER and KGR. Fig. \ref{KERKGR_SNR_quantization_factor} compares the performance of KER and KGR under the quantization factor is 0.005, 0.1, 0.2 and 0.3. The outcomes show that when $\varepsilon$ is higher, the KER becomes smaller and the KGR becomes larger. In practice, to satisfy the high agreement, the KER need to be less than $10^{-3}$. To reduce the KER as well as the information exposed during the information reconciliation, different quantization factors can be set under different SNRs. For example, we can set $\varepsilon$ is 0.1 when SNR exceeds 25 dB, $\varepsilon$ is 0.2 when SNR is between 15 dB to 25 dB, $\varepsilon$ is 0.3 when SNR is under 15 dB.

The NIST is a common tool to evaluate the randomness feature of binary sequences \cite{8314118}, which is also adopted in our work.
The output of each test is the p-value. When the p-value is greater than a threshold, which is usually 0.01, the detected sequence passes the test. 
Since the quantization method we adopted was discarded in the quantization process, the length of each group of keys cannot reach the minimum requirement of 128 bits for NIST detection. We splice 20000 sets of keys together to form $t$ sets of 128-bit keys, generally $t$ is less than 20000.
We perform 8 NIST statistical tests on $t=15,981$ sets of keys generated under 25 dB. The serial test includes two types of tests; the serial test is deemed passed when both tests pass. Table \ref{table_randomness} shows the pass rate of the test with different quantization intervals, which is the ratio of the number of sets passed to the number of all sets. 
\begin{table}[!t] 
\caption{NIST statistical test pass ratio.} 
	\centering
	\begin{tabular}{|p{4cm}<{\centering}| p{3cm}<{\centering} |}
		\hline
		 \bfseries Test&\bfseries Pass ratio \\ \hline
		 Approximate Entropy & 0.8363 \\\hline
		 Block Frequency & 0.9158\\\hline
		 Cumulative Sums &0.9972 \\\hline
		 Discrete Fourier Transform & 0.9995 \\\hline
		 Frequency & 0.7103\\\hline
		 Ranking & 0.8511\\\hline
		 Runs & 0.8733 \\\hline
		 Serial & 0.7461\\
		 \hline
	\end{tabular}
\label{table_randomness}
\end{table}

\subsubsection{The Performance at Eve}

The performance of Eve is related to whether the key can be eavesdropped. Since Eve is close to Bob and Bob receives the channel on $f_{AB}$, it is assumed that Eve eavesdrops on the channel over frequency $f_{AB}$ without loss of generality.
Alice, Bob and Eve can have channel feature vectors $\mathbf{x}_A$, $\mathbf{x}_B$ and $\mathbf{x}_E$ respectively. Fig.~\ref{fig:hor_2figs_1cap_2subcap_1} shows the channel features obtained by Alice, Bob and Eve with Bob at position 1 and Eve at position 2. The distance between position 1 and position 2 is 0.15 m, which just satisfies (\ref{d}) .The channel features obtained by Alice and Bob come from two channels with the same relative position and different frequencies, while those obtained by Bob and Eve come from two channels with different positions and the same frequencies. As Fig. \ref{fig:hor_2figs_1cap_2subcap_1} shown, changes in location and frequency both cause changes in channel characteristics. Fig.~\ref{fig:hor_2figs_1cap_2subcap_2} shows the channel features obtained by Alice, Bob and Eve after Alice used KGNet for band feature mapping. The channel features obtained by Alice and Bob have a high degree of reciprocity, which are greatly different from those obtained by Eve.
\begin{figure}
	\centering
	\subfigure[Before Alice uses KGNet for channel feature mapping. SNR is 25 dB.]{
		\begin{minipage}[b]{\linewidth}
			\includegraphics[width=\linewidth]{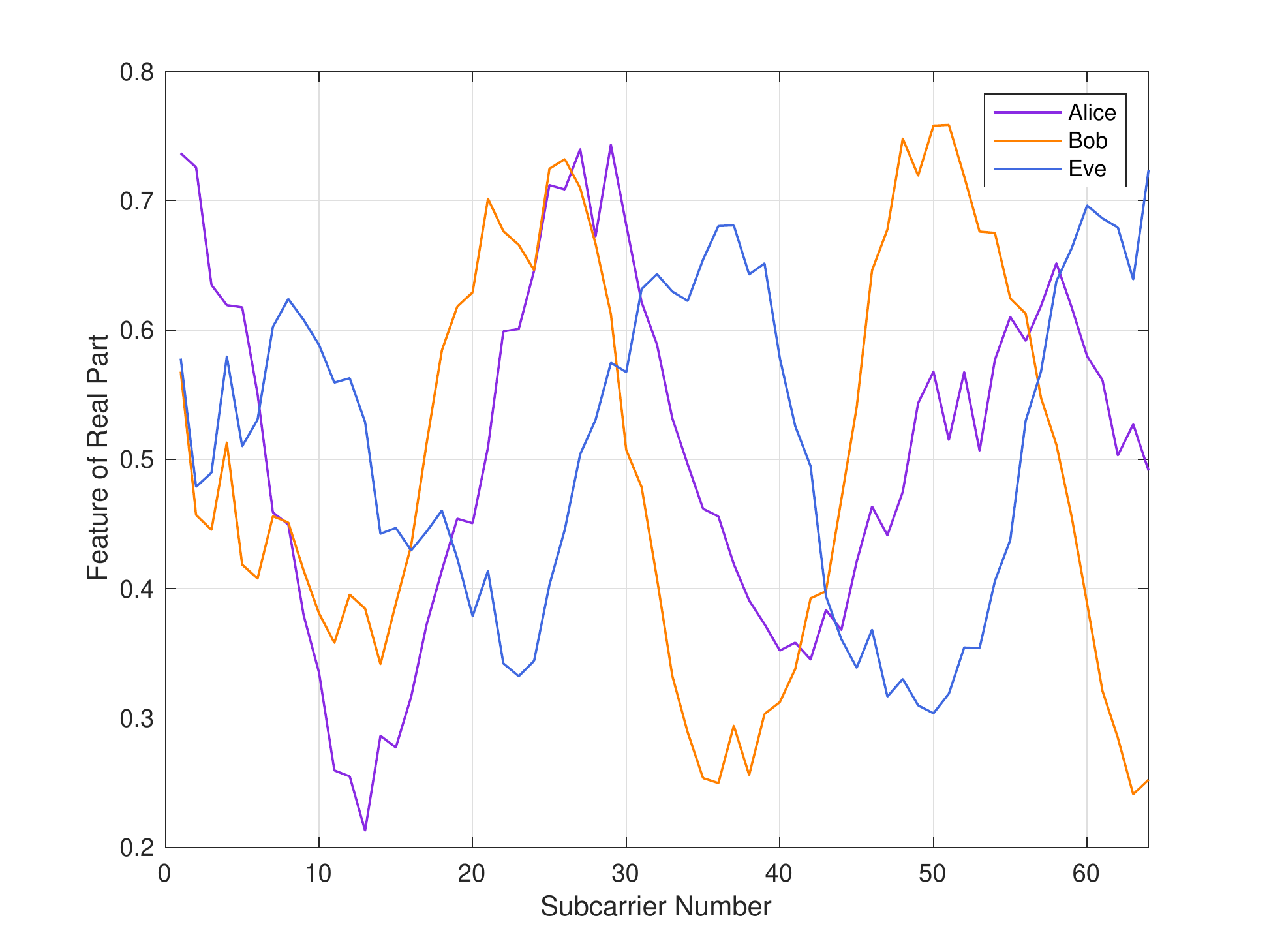}
		\end{minipage}
		\label{fig:hor_2figs_1cap_2subcap_1}
	}
	\subfigure[After Alice uses KGNet for channel feature mapping. SNR is 25 dB.]{
		\begin{minipage}[b]{\linewidth}
			\includegraphics[width=\linewidth]{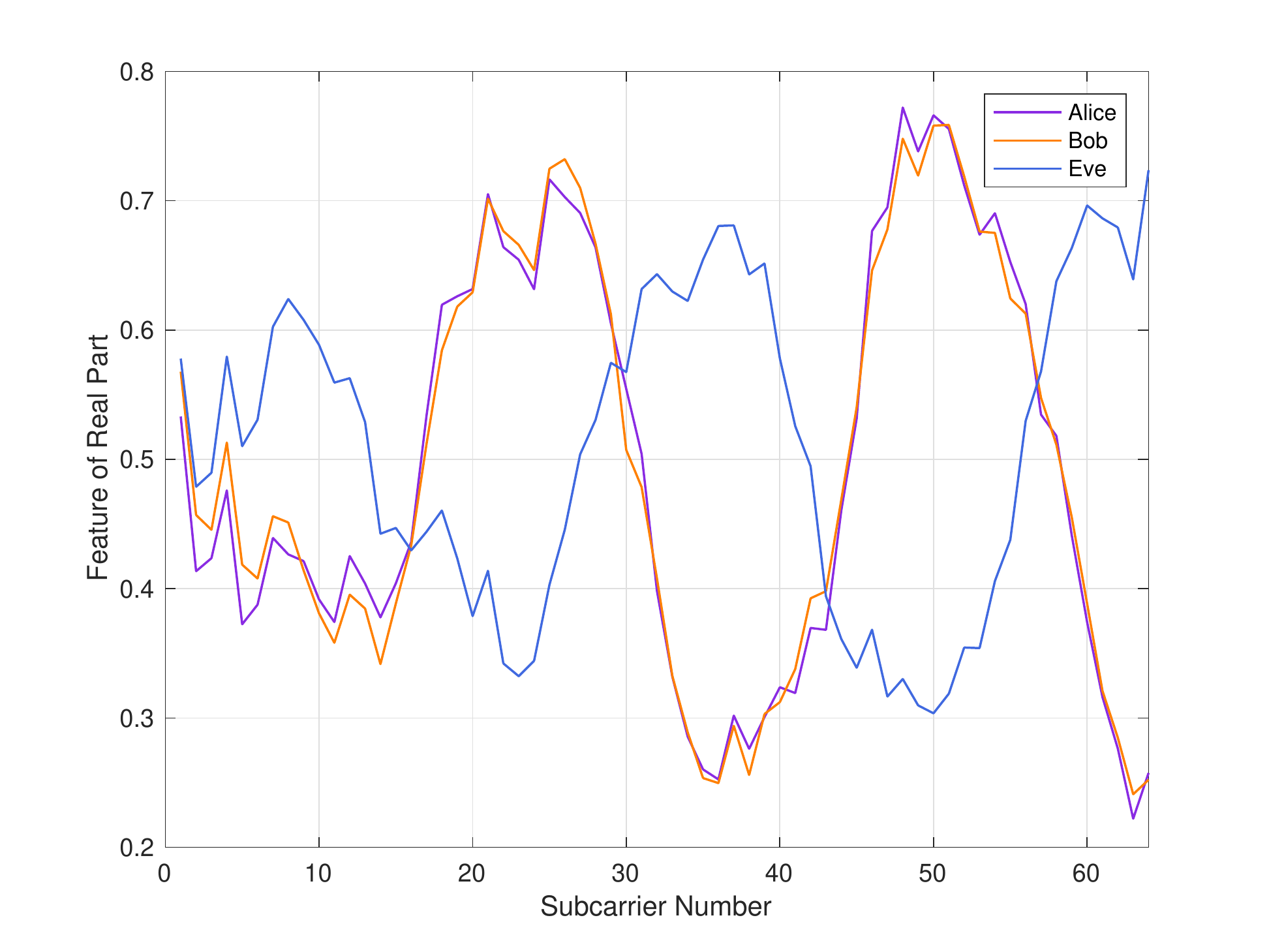}
		\end{minipage}
		\label{fig:hor_2figs_1cap_2subcap_2}
	}
	\caption{The performance of Alice, Bob and Eve.}
	\label{fig:hor_2figs_1cap_2subcap}
\end{figure}

 Fig. \ref{NMSE_SNR_Eve} compares the NVD performance of channel features obtained by Alice and Bob before and after using KGNet for band feature mapping and between Eve and Bob under different SNRs, when Bob and Eve are fixed at positions 1 and 2, respectively. 
We can find that when Alice does not use KGNet for feature mapping, the NVD of channel features obtained by Alice and Bob and the NVD of channel features obtained by Eve and Bob are similar. However, after Alice uses KGNet, the NVD between the channel feature obtained by Alice and Bob are much smaller than those obtained by Eve and Bob. Therefore, we believe that even when Eve and Bob are close in the same room, Eve has limited access to the channel features that Alice and Bob can use to generate keys.
\begin{figure}[!t]
	\centering 
	\includegraphics[width=\linewidth]{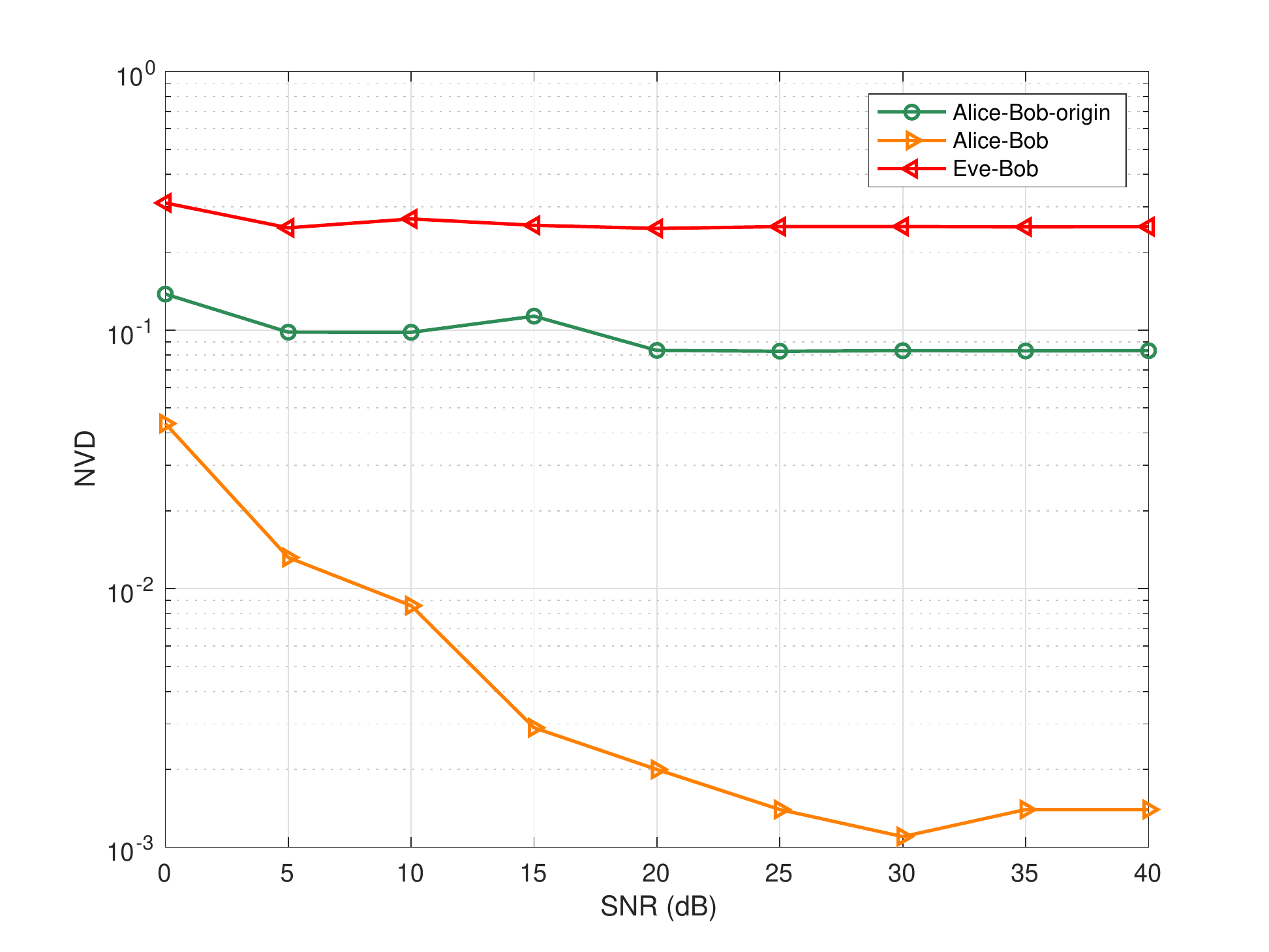} 
	\caption{The NVD of the channel features between Alice and Bob originally, between Alice and Bob after Alice uses KGNet, and between Eve and Bob versus SNR when Bob and Eve are fixed in position 1 and position 2, respectively. }
	\label{NMSE_SNR_Eve} 
\end{figure}

\subsection{Overhead Analysis}
\begin{table*}[!t] 
\caption{Complexity analysis and time cost for the four networks.}
	\centering
	\begin{tabular}{p{3cm}<{\centering}| p{2.5cm}<{\centering} | p{2.5cm}<{\centering} |p{2.5cm}<{\centering}|p{2.5cm}<{\centering}|p{2.5cm}<{\centering}}
		\hline
		 \bfseries Network&\bfseries Training time & \bfseries CPU average load & \bfseries GPU memory utilization &\bfseries FLOPs & \bfseries The number of trained parameters\\ \hline
		 FNN in \cite{huang2019deep}&12min 13s &15.2\%& 4.8 / 9.9 GB&2,357,120 &1,181,824\\
		 Advanced-FNN in \cite{huang2019deep}&12min 25s &15.6\%&4.8 / 9.9 GB&2,357,120 &1,181,824\\
		 KGNet & \textbf{14min 42s} & 15.2\% &4.8 / 9.9 GB &  \textbf{4,453,248}& \textbf{2,231,424} \\
		 1D-CNN &14min~~2s&\textbf{18.0\%}&\textbf{5.1 / 9.9 GB}&2,850,688&1,362,120\\
		 \hline
	\end{tabular}
\label{table_complexity}
\end{table*}
To measure the complexity of the KGNet, multiple indicators are used in this paper. Table \ref{table_complexity} compares the complexity of the four networks in terms of five indicators, i.e., training time, CPU average load, GPU memory utilization, the number of float-point operations (FLOPs), and the number of trained parameters. We calculate the number of FLOPs with reference \cite{molchanov2017pruning}. 
In terms of training time, it takes about 14min 42s to train the KGNet, which takes longer than other networks. 
In terms of the numbers of FLOPs and trained parameters, the numbers of FLOPs and training parameters required by the KGNet are nearly twice that of other networks. 
In terms of  CPU and GPU cost, the training of the KGNet requires approximately 15.2\% of the CPU load and 4.8 G of GPU memory (the total GPU memory is 9.9 G), which is similar to the cost of other types of networks.
In general, the resource consumption required for the KGNet training is higher than that of other networks.
However, this overhead is considered acceptable, and the performance improvement brought about by these excess consumption is more important. The reasons are summarized as follows.
\begin{itemize}
	\item We choose to train and deploy the network at the base station, and the base station usually has a large amount of computing resources, which is sufficient to meet the resources required for network training and storage. The training time required to train the network on the base station will be greatly reduced than in the simulation. Besides, the training of the neural network can be done on the cloud without spending any resources of the base station.
	\item We are modeling based on location, and only need to train a network for a large area. All terminal devices in this area can establish a secure connection with the base station by using this network. As long as the environment in the area does not undergo large-scale changes, the trained network can be used forever. Even if the environment changes, we can fine-tune the network through a small number of datasets.
	\item The networks with better performance can generate initial keys with lower BER and higher BGR, which will reduce the overhead of subsequent information reconciliation. Therefore, from a long-term perspective, it is very cost-effective to exchange the additional cost required for training for a network with better performance.
\end{itemize}

Besides, it needs to be emphasized that compared with the key generation for TDD systems,  the base station does not have any additional overhead except for the additional use of the KGNet for frequency band feature mapping in FDD systems. In particular, the terminal does not need to train and save the neural network, and there is no additional overhead other than regular key generation, which means it is suitable for resource-constrained IoT devices.

\section{Conclusion}
\label{conclusion}This paper is the first work to apply deep learning to the design of FDD key generation scheme. We first demonstrated the mapping of channel features in different frequency bands under a condition that the mapping function from the candidate user positions to the channels is bijective. Then, we proposed the KGNet for frequency band feature mapping to construct reciprocal channel feature between communication parties. Based on the KGNet, a novel secret key generation scheme for FDD systems was established. Simulation results have demonstrated that the KGNet performs better than other benchmark neural networks in terms of both fitting and generalization under low SNR. In addition, training the neural network with a cross dataset that combines noisy dataset and noise-free dataset can improve the robustness of the neural network. Moreover, the proposed KGNet-based FDD system key generation scheme is evaluated from the KER, KGR and randomness, and it is verified that it can be used for key generation for FDD systems. Our future work will extend this approach to design key generation for FDD-MIMO systems.

\appendix
\section{Appendix}
\subsection{Proof of Proposition 1}
\label{proof1}
\begin{proof}
	Under definition 1, we have the mapping $\boldsymbol{\Phi}_{f_{AB}}:\{P_B\}\rightarrow\{\textbf{H}(f_{AB})\}$. Under assumption 1, we have the mapping $\boldsymbol{\Phi}_{f_{BA}}^{-1}:\{\textbf{H}(f_{BA})\} \rightarrow \{P_B\}$. Since the domain of $\boldsymbol{\Phi}_{f_{AB}}$ is the same as co-domain of $\boldsymbol{\Phi}_{f_{BA}}^{-1}$, the composite mapping $\boldsymbol{\Psi}_{f_{BA}\rightarrow f_{AB}}$ exists.
\end{proof}

\subsection{Proof of Proposition 2}
\label{proof2}
\begin{proof}
	Under proposition 1, we have the mapping $\boldsymbol{\Psi}_{f_{BA}\rightarrow f_{AB}}$. Under mapping function $\boldsymbol{\xi}_{f}$ and $\boldsymbol{\xi}_{f_{AB}}^{-1}$, we have the mapping $\boldsymbol{\xi}_{f_{AB}}:\textbf{H}(f_{AB})\rightarrow\textbf{x}(f_{AB})$ and $\boldsymbol{\xi}_{f_{BA}}^{-1}:\textbf{x}(f_{BA}) \rightarrow \textbf{H}(f_{BA})$. Since the domain of $\boldsymbol{\xi}_{f_{AB}}$ is the same as co-domain of $\boldsymbol{\Psi}_{f_{BA}\rightarrow f_{AB}}$, the domain of $\boldsymbol{\Psi}_{f_{BA}\rightarrow f_{AB}}$ is the same as co-domain of $\boldsymbol{\xi}_{f_{BA}}^{-1}$, the mapping $\boldsymbol{\Psi}^{'}_{f_{BA}\rightarrow f_{AB}}$ exists.
\end{proof}
\subsection{Proof of Theorem 1}
\label{proof3}
\begin{proof}
	(i) Since $\textbf{h}_{f_{BA}}$ is bounded and closed, $\textbf{x}_{f_{BA}}$ obtained after linear change is still bounded and closed, $\mathbb{H}$ is a compact set; (ii) Since $\boldsymbol{\Phi}_{f_{AB}}$ and $\boldsymbol{\Phi}_{f_{BA}}^{-1}$ are continuous mappig and the composition of continuous mappings is still a continuous mappinng, we know $\boldsymbol{\Psi}_{f_{BA}\rightarrow f_{AB}}$ is a continuous function. (iii) Since  $\boldsymbol{\Psi}_{f_{BA}\rightarrow f_{AB}}$and $\boldsymbol{\xi}$ are a continuous function, we know $\boldsymbol{\Psi}^{'}_{f_{BA}\rightarrow f_{AB}}(\textbf{x}(f_{BA}))$ is also a continuous function for $ \forall \textbf{x}(f_{BA}) \in \mathbb{H}$. Based on (i), (ii), (iii) and universal approximation theorem \cite[Theorem 1]{hornik1989multilayer}, Theorem 1 is proved.
\end{proof}

\bibliographystyle{IEEEtran} 
\bibliography{IEEEabrv,myref}

\begin{thebibliography}{10}
\providecommand{\url}[1]{#1}
\csname url@samestyle\endcsname
\providecommand{\newblock}{\relax}
\providecommand{\bibinfo}[2]{#2}
\providecommand{\BIBentrySTDinterwordspacing}{\spaceskip=0pt\relax}
\providecommand{\BIBentryALTinterwordstretchfactor}{4}
\providecommand{\BIBentryALTinterwordspacing}{\spaceskip=\fontdimen2\font plus
\BIBentryALTinterwordstretchfactor\fontdimen3\font minus
  \fontdimen4\font\relax}
\providecommand{\BIBforeignlanguage}[2]{{%
\expandafter\ifx\csname l@#1\endcsname\relax
\typeout{** WARNING: IEEEtran.bst: No hyphenation pattern has been}%
\typeout{** loaded for the language `#1'. Using the pattern for}%
\typeout{** the default language instead.}%
\else
\language=\csname l@#1\endcsname
\fi
#2}}
\providecommand{\BIBdecl}{\relax}
\BIBdecl

\bibitem{zou2016survey}
Y.~Zou, J.~Zhu, X.~Wang, and L.~Hanzo, ``{A survey on wireless security:
  Technical challenges, recent advances, and future trends},'' \emph{Proc.
  IEEE}, vol. 104, no.~9, pp. 1727--1765, Sep. 2016.

\bibitem{william2019cryptography}
W.~Stallings, \emph{{Cryptography and Network Security: Principles and
  Practice}}.\hskip 1em plus 0.5em minus 0.4em\relax Prenntice Hall, 2011.

\bibitem{wang2019physical}
N.~Wang, P.~Wang, A.~Alipour-Fanid, L.~Jiao, and K.~Zeng, ``{Physical-layer
  security of 5G wireless networks for IoT: Challenges and opportunities},''
  \emph{IEEE Internet of Things J.}, vol.~6, no.~5, pp. 8169--8181, Jul. 2019.

\bibitem{10.1145/3140257}
J.~Wan, A.~Lopez, and M.~A.~A. Faruque, ``Physical layer key generation:
  Securing wireless communication in automotive cyber-physical systems,''
  \emph{ACM Trans. Cyber-Phys. Syst.}, vol.~3, no.~2, Oct. 2018.

\bibitem{zhang2016key}
J.~Zhang, T.~Q. Duong, A.~Marshall, and R.~Woods, ``{Key generation from
  wireless channels: A review},'' \emph{{IEEE Access}}, vol.~4, pp. 614--626,
  Jan. 2016.

\bibitem{zhang2020new}
J.~Zhang, G.~Li, A.~Marshall, A.~Hu, and L.~Hanzo, ``{A new frontier for IoT
  security emerging from three decades of key generation relying on wireless
  channels},'' \emph{IEEE Access}, vol.~8, pp. 138\,406--138\,446, Jul. 2020.

\bibitem{Li2019physical}
G.~Li, C.~Sun, J.~Zhang, E.~Jorswieck, B.~Xiao, and A.~Hu, ``{Physical layer
  key generation in 5G and beyond wireless communications: Challenges and
  opportunities},'' \emph{Entropy}, vol.~21, no.~5, p. 497, May 2019.

\bibitem{wang2012wireless}
W.~Wang, H.~Jiang, X.~Xia, P.~Mu, and Q.~Yin, ``{A wireless secret key
  generation method based on Chinese remainder theorem in FDD systems},''
  \emph{Sci. China Inf. Sci.}, vol.~55, no.~7, pp. 1605--1616, Jul. 2012.

\bibitem{liu2019secret}
B.~Liu, A.~Hu, and G.~Li, ``{Secret key generation scheme based on the channel
  covariance matrix eigenvalues in FDD systems},'' \emph{IEEE Commun. Lett.},
  vol.~23, no.~9, pp. 1493--1496, Sep. 2019.

\bibitem{goldberg2013method}
S.~J. Goldberg, Y.~C. Shah, and A.~Reznik, ``{Method and apparatus for
  performing JRNSO in FDD, TDD and MIMO communications},'' U.S. Patent 8 401
  196 B2, Mar. 19, 2013.

\bibitem{wu2013secret}
X.~Wu, Y.~Peng, C.~Hu, H.~Zhao, and L.~Shu, ``{A secret key generation method
  based on CSI in OFDM-FDD system},'' in \emph{Proc. IEEE Globecom Workshops.
  (GC Wkshps)}, Atlanta, GA, United states, Dec. 2013, pp. 1297--1302.

\bibitem{qin2016exploiting}
D.~Qin and Z.~Ding, ``{Exploiting multi-antenna non-reciprocal channels for
  shared secret key generation},'' \emph{IEEE Trans.Inf. Forensics Secur.},
  vol.~11, no.~12, pp. 2693--2705, Dec. 2016.

\bibitem{allam2017channel}
A.~M. Allam, ``{Channel-based secret key establishment for FDD wireless
  communication systems},'' \emph{Commun. Appl. Electron}, vol.~7, no.~9, pp.
  27--31, Nov. 2017.

\bibitem{li2018constructing}
G.~Li, A.~Hu, C.~Sun, and J.~Zhang, ``{Constructing reciprocal channel
  coefficients for secret key generation in FDD systems},'' \emph{IEEE Commun.
  Lett.}, vol.~22, no.~12, pp. 2487--2490, Dec. 2018.

\bibitem{alrabeiah2019deep}
M.~Alrabeiah and A.~Alkhateeb, ``{Deep learning for TDD and FDD massive MIMO:
  Mapping channels in space and frequency},'' in \emph{Proc. 53rd Asilomar
  Conf. Rec. Asilomar Conf. Signals Syst. Comput. (ACSSC)}, Pacific Grove, CA,
  United states, Nov. 2019, pp. 1465--1470.

\bibitem{vasisht2016eliminating}
D.~Vasisht, S.~Kumar, H.~Rahul, and D.~Katabi, ``{Eliminating Channel Feedback
  in Next-Generation Cellular Networks},'' in \emph{Proc. ACM SIGCOMM},
  Florianopolis, Brazil, Aug. 2016, p. 398–411.

\bibitem{linning2018investigation}
L.~Peng, G.~Li, J.~Zhang, R.~Woods, M.~Liu, and A.~Hu, ``{An investigation of
  using loop-back mechanism for channel reciprocity enhancement in secret key
  generation},'' \emph{IEEE Trans Mob Comput}, vol.~18, no.~3, pp. 507--519,
  May 2018.

\bibitem{yang2019deep1}
Y.~Yang, F.~Gao, X.~Ma, and S.~Zhang, ``{Deep learning-based channel estimation
  for doubly selective fading channels},'' \emph{IEEE Access}, vol.~7, pp.
  36\,579--36\,589, Mar. 2019.

\bibitem{wen2018deep}
C.-K. Wen, W.-T. Shih, and S.~Jin, ``{Deep Learning for Massive MIMO CSI
  Feedback},'' \emph{IEEE Wireless Commun. Lett}, vol.~7, no.~5, pp. 748--751,
  Oct. 2018.

\bibitem{wang2019ul}
J.~Wang, Y.~Ding, S.~Bian, Y.~Peng, M.~Liu, and G.~Gui, ``{UL-CSI data driven
  deep learning for predicting DL-CSI in cellular FDD systems},'' \emph{IEEE
  Access}, vol.~7, pp. 96\,105--96\,112, Jul. 2019.

\bibitem{lin2020improved}
Y.~Lin, Y.~Tu, and Z.~Dou, ``{An Improved Neural Network Pruning Technology for
  Automatic Modulation Classification in Edge Devices},'' \emph{IEEE Trans.
  Veh. Technol.}, vol.~69, no.~5, pp. 5703--5706, Mar. 2020.

\bibitem{gao2019deep}
S.~Gao, P.~Dong, Z.~Pan, and G.~Y. Li, ``{Deep Learning Based Channel
  Estimation for Massive MIMO With Mixed-Resolution ADCs},'' \emph{IEEE Commun.
  Lett.}, vol.~23, no.~11, pp. 1989--1993, Aug. 2019.

\bibitem{yang2019deep}
Y.~Yang, F.~Gao, G.~Y. Li, and M.~Jian, ``{Deep Learning-Based Downlink Channel
  Prediction for FDD Massive MIMO System},'' \emph{IEEE Commun. Lett.},
  vol.~23, no.~11, pp. 1994--1998, Nov. 2019.

\bibitem{safari2019deep}
M.~S. Safari, V.~Pourahmadi, and S.~Sodagari, ``{Deep UL2DL: Data-Driven
  Channel Knowledge Transfer From Uplink to Downlink},'' \emph{IEEE Open J.
  Veh. Technol.,}, vol.~1, no.~5, pp. 29--44, Dec. 2020.

\bibitem{mathur2008radio}
S.~Mathur, W.~Trappe, N.~Mandayam, C.~Ye, and A.~Reznik, ``{Radio-telepathy:
  extracting a secret key from an unauthenticated wireless channel},'' in
  \emph{Proc. Annu Int Conf Mobile Comput Networking (Mobicom)}, San Francisco,
  CA, United states, Sep. 2008, pp. 128--139.

\bibitem{aono2005wireless}
T.~Aono, K.~Higuchi, T.~Ohira, B.~Komiyama, and H.~Sasaoka, ``Wireless secret
  key generation exploiting reactance-domain scalar response of multipath
  fading channels,'' \emph{IEEE Trans. Antennas Propag.}, vol.~53, no.~11, pp.
  3776--3784, Nov. 2005.

\bibitem{xu2018lora}
W.~Xu, S.~Jha, and W.~Hu, ``{LoRa-Key: Secure key generation system for
  LoRa-based network},'' \emph{IEEE Internet of Things J.}, vol.~6, no.~4, pp.
  6404--6416, Aug. 2019.

\bibitem{ruotsalainen2019experimental}
H.~Ruotsalainen, J.~Zhang, and S.~Grebeniuk, ``{Experimental Investigation on
  Wireless Key Generation for Low-Power Wide-Area Networks},'' \emph{IEEE
  Internet of Things J.}, vol.~7, no.~3, pp. 1745--1755, Oct. 2019.

\bibitem{8314118}
G.~Li, A.~Hu, J.~Zhang, L.~Peng, C.~Sun, and D.~Cao, ``{High-Agreement
  Uncorrelated Secret Key Generation Based on Principal Component Analysis
  Preprocessing},'' \emph{IEEE Trans Commun}, vol.~66, no.~7, pp. 3022--3034,
  Jul. 2018.

\bibitem{9123376}
G.~{Li}, Z.~{Zhang}, J.~{Zhang}, and A.~{Hu}, ``{Encrypting Wireless
  Communications on the Fly Using One-Time Pad and Key Generation},''
  \emph{IEEE Internet of Things J.}, vol.~8, no.~1, pp. 357--369, Jun. 2021.

\bibitem{zhang2016efficient}
J.~Zhang, A.~Marshall, R.~Woods, and T.~Q. Duong, ``{Efficient key generation
  by exploiting randomness from channel responses of individual OFDM
  subcarriers},'' \emph{IEEE Trans Commun}, vol.~64, no.~6, pp. 2578--2588,
  Apr. 2016.

\bibitem{vieira2017deep}
J.~Vieira, E.~Leitinger, M.~Sarajlic, X.~Li, and F.~Tufvesson, ``{Deep
  convolutional neural networks for massive MIMO fingerprint-based
  positioning},'' in \emph{Proc. IEEE Int Symp Person Indoor Mobile Radio
  Commun. (PIMRC)}, Montreal, QC, Canada, Oct. 2017, pp. 1--6.

\bibitem{hornik1989multilayer}
K.~Hornik, M.~Stinchcombe, H.~White \emph{et~al.}, ``{Multilayer feedforward
  networks are universal approximators},'' \emph{Neural Netw}, vol.~2, no.~5,
  pp. 359--366, 1989.

\bibitem{kingma2014adam}
\BIBentryALTinterwordspacing
D.~P. Kingma and J.~Ba, ``{Adam: A method for stochastic optimization},'' 2014.
  [Online]. Available: \url{http://arxiv.org/abs/1412.6980}
\BIBentrySTDinterwordspacing

\bibitem{zenger2015security}
C.~Zenger, J.~Zimmer, and C.~Paar, ``{Security Analysis of Quantization Schemes
  for Channel-Based Key Extraction},'' in \emph{Proc. Workshop Wireless Commun.
  Security Phys. Layer}, Coimbra, Portugal, Jul. 2015, pp. 267--272.

\bibitem{zhu2013extracting}
X.~Zhu, F.~Xu, E.~Novak, C.~C. Tan, Q.~Li, and G.~Chen, ``{Extracting secret
  key from wireless link dynamics in vehicular environments},'' in \emph{Proc.
  IEEE INFOCOM}, Turin, Italy, Apr. 2013, pp. 2283--2291.

\bibitem{dodis2004fuzzy}
Y.~Dodis, L.~Reyzin, and A.~Smith, ``{Fuzzy extractors: How to generate strong
  keys from biometrics and other noisy data},'' in \emph{Proc. Int'l Conf.
  Theory and Applications of Cryptographic Techniques}, May 2004, pp. 523--540.

\bibitem{alkhateeb2019deepmimo}
A.~Alkhateeb, ``{{DeepMIMO}: A Generic Deep Learning Dataset for Millimeter
  Wave and Massive {MIMO} Applications},'' in \emph{Proc. Inf. Theory Appl.
  Workshop (ITA)}, San Diego, CA, USA, Feb. 2019, pp. 1--8.

\bibitem{Remcom}
\BIBentryALTinterwordspacing
``Remcom wireless insite.'' [Online]. Available:
  \url{http://www.remcom.com/wireless-insite}
\BIBentrySTDinterwordspacing

\bibitem{rukhin2001statistical}
A.~Rukhin, J.~Soto, J.~Nechvatal, M.~Smid, and E.~Barker, ``A statistical test
  suite for random and pseudorandom number generators for cryptographic
  applications,'' Booz-allen and hamilton inc mclean va, Tech. Rep., 2001.

\bibitem{huang2019deep}
C.~Huang, G.~C. Alexandropoulos, A.~Zappone, C.~Yuen, and M.~Debbah, ``{Deep
  learning for UL/DL channel calibration in generic massive MIMO systems},'' in
  \emph{Proc. IEEE Int Conf Commun. (ICC)}, Shanghai, China, May 2019, pp.
  1--6.

\bibitem{molchanov2017pruning}
\BIBentryALTinterwordspacing
P.~Molchanov, S.~Tyree, T.~Karras, T.~Aila, and J.~Kautz, ``Pruning
  convolutional neural networks for resource efficient inference,'' 2017.
  [Online]. Available: \url{http://arxiv.org/abs/1611.06440}
\BIBentrySTDinterwordspacing

\end{thebibliography}


\begin{IEEEbiography}[{\includegraphics[width=1in,height=1.25in,clip,keepaspectratio]{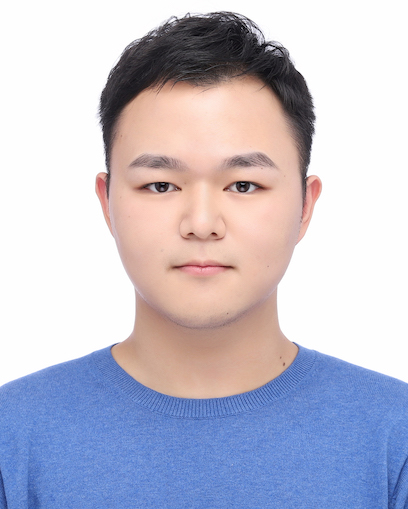}}]{Xinwei Zhang} (Student Member, IEEE) received the B.S. degree in computer science and technology from Henan University of Technology, Zhengzhou, China, in 2018. He is currently pursuing the M.S. degree in computer technology with Southeast University. His research interests include physical-layer security, and secret key generation.
\end{IEEEbiography}

\begin{IEEEbiography}[{\includegraphics[width=1in,height=1.25in,clip,keepaspectratio]{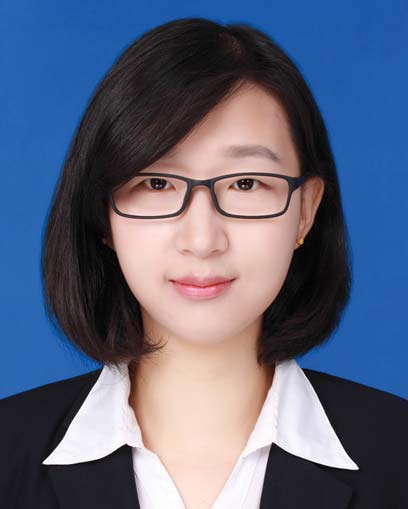}}]{Guyue Li}
(Member, IEEE) received the B.S. degree in information science and technology and the Ph.D. degree in information security from Southeast University, Nanjing, China, in 2011 and 2017, respectively. 

From June 2014 to August 2014, she was a Visiting Student with the Department of Electrical Engineering, Tampere University of Technology, Finland. She is currently an Associate Professor with the School of Cyber Science and Engineering, Southeast University. Her research interests include physical-layer security, secret key generation, radio frequency fingerprint, and link signature.
\end{IEEEbiography}

\begin{IEEEbiography}[{\includegraphics[width=1in,height=1.25in,clip,keepaspectratio]{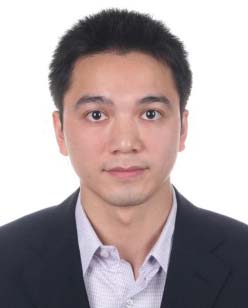}}]{Junqing Zhang}
received the B.Eng and M.Eng degrees in Electrical Engineering from Tianjin University, China in 2009 and 2012, respectively, and the Ph.D degree in Electronics and Electrical Engineering from Queen's University Belfast, UK in 2016. 
From Feb. 2016 to Jan. 2018, he was a Postdoctoral Research Fellow with Queen's University Belfast. From Feb. 2018 to May 2020, he was a Tenure Track Fellow (Assistant Professor) with University of Liverpool, UK. Since June 2020, he is a Lecturer (Assistant Professor) with University of Liverpool. His research interests include Internet of Things, wireless security, physical layer security, key generation, and radio frequency fingerprint identification.

\end{IEEEbiography}

\begin{IEEEbiography}[{\includegraphics[width=1in,height=1.25in,clip,keepaspectratio]{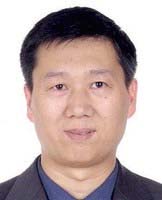}}]{Aiqun Hu}
(Senior Member, IEEE) received the B.Sc.(Eng.), M.Eng.Sc., and Ph.D. degrees from Southeast University in 1987, 1990, and 1993, respectively. 

He was invited as a Post-Doctoral Research Fellow with The University of Hong Kong from 1997 to 1998, and a TCT Fellow with Nanyang Technological University in 2006. He has published two books and more than 100 technical articles in wireless communications field. His research interests include data transmission and secure communication technology.
\end{IEEEbiography}

\begin{IEEEbiography}[{\includegraphics[width=1in,height=1.25in,clip,keepaspectratio]{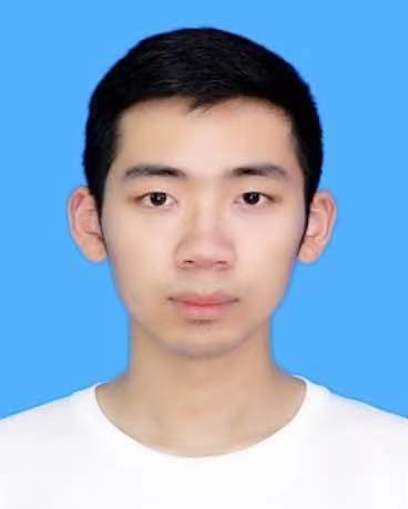}}]{Zongyue Hou}
received the B.S. degree in communication engineering from the Jiangnan University, Wuxi, China, in 2020. He is currently pursuing the M.S. degree in electronic information with Southeast University. His research interests include physical-layer security, and secret key generation.
\end{IEEEbiography}

\begin{IEEEbiography}[{\includegraphics[width=1in,height=1.25in,clip,keepaspectratio]{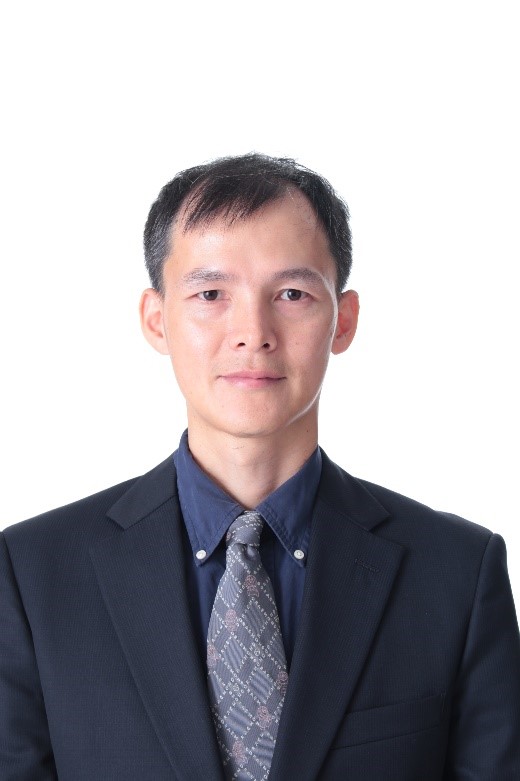}}]{Bin Xiao}
(S’01-M’04-SM’11) received the B.Sc. and M.Sc. degrees in electronics engineering from Fudan University, China, and the Ph.D. degree in computer science from The University of Texas at Dallas, USA. He is currently a Professor with the Department of Computing, The Hong Kong Polytechnic University. Dr. Xiao has over ten years research experience in the cyber security, and currently focuses on the blockchain technology and AI security. He published more than 100 technical papers in international top journals and conferences. Currently, he is the associate editor of the Journal of Parallel and Distributed
Computing (JPDC) and the vice chair of IEEE ComSoc CISTC committee. He has been the symposium co-chair of IEEE ICC2020, ICC 2018 and Globecom 2017, and the general chair of IEEE SECON 2018.
\end{IEEEbiography}

\end{document}